\documentclass[a4paper,UKenglish,cleveref, autoref]{lipics-v2019}
\usepackage{amsmath,amstext,amsthm,amssymb,amsfonts}
\usepackage{graphicx}
\usepackage{comment} 
\usepackage{tikz}
\usepackage{tkz-graph}
\tikzstyle{vertex}=[circle, draw, inner sep=0pt, minimum size=6pt]
\newcommand{\vertex}{\node[vertex]}

\newcommand{\nphard}{NP-hard}

\newcommand{\mcD}{\mathcal{D}}
\newcommand{\mcU}{\mathcal{U}}

\newcommand{\mcF}{\mathcal{F}}
\newcommand{\mcN}{\mathcal{N}}
\newcommand{\mcP}{\mathcal{P}}
\newcommand{\Mplus}{2^{[m]}\setminus\emptyset}

\title{The Complexity of Partial Function Extension for Coverage Functions} 

\author{Umang Bhaskar}{Tata Institute of Fundamental Research, Mumbai, India}{umang.bhaskar@tifr.res.in}{}{Supported in part by a Ramanujan fellowship and an Early Career Research award.}

\author{Gunjan Kumar}{Tata Institute of Fundamental Research, Mumbai, India}{gunjan.kumar@tifr.res.in}{}{}

\authorrunning{U. Bhaskar and G. Kumar}



\relatedversion{}

\nolinenumbers 

\hideLIPIcs  

\EventEditors{Dimitris Achlioptas and L\'{a}szl\'{o} A. V\'{e}gh}
\EventNoEds{2}
\EventLongTitle{Approximation, Randomization, and Combinatorial Optimization. Algorithms and Techniques (APPROX/RANDOM 2019)}
\EventShortTitle{\mbox{\scriptsize APPROX/RANDOM 2019}}
\EventAcronym{APPROX/RANDOM}
\EventYear{2019}
\EventDate{September 20--22, 2019}
\EventLocation{Massachusetts Institute of Technology, Cambridge, MA, USA}
\EventLogo{}
\SeriesVolume{145}
\ArticleNo{30}

\begin{document}

\maketitle

\begin{abstract}
Coverage functions are an important subclass of submodular functions, finding applications in machine learning, game theory, social networks, and facility location. We study the complexity of partial function extension to coverage functions. That is, given a partial function consisting of a family of subsets of $[m]$ and a value at each point, does there exist a coverage function defined on all subsets of $[m]$ that extends this partial function? Partial function extension is previously studied for other function classes, including boolean functions and convex functions, and is useful in many fields, such as obtaining bounds on learning these function classes.

We show that determining extendibility of a partial function to a coverage function is NP-complete, establishing in the process that there is a polynomial-sized certificate of extendibility. The hardness also gives us a lower bound for learning coverage functions. We then study two natural notions of approximate extension, to account for errors in the data set. The two notions correspond roughly to multiplicative point-wise approximation and additive $L_1$ approximation. We show upper and lower bounds for both notions of approximation. In the second case we obtain nearly tight bounds.
\end{abstract}

	\section{Introduction}

When can a \emph{partial function} --- given as a set of points from a domain, and a value at each point --- be extended to a \emph{total function} on the domain, that lies in some particular class of functions? This is the basic question of \emph{partial function extension}, and is studied both independently (such as in convex analysis) and as a recurring subproblem in many areas in combinatorial optimization, including computational learning and property testing. 

In this paper we study the computational complexity of partial function extension for \emph{coverage functions}. Coverage functions are a natural and widely-studied subclass of submodular functions that find many applications, including in machine learning~\cite{krause2008robust}, auctions~ \cite{blumrosen2007combinatorial,lehmann2006combinatorial}, influence maximization~\cite{borgs2014maximizing,seeman2013adaptive}, and plant location~\cite{cornuejols1977exceptional}. For a natural number $m$, let $[m]$ denote the set $\{1,2,\ldots, m\}$. A set function $f:2^{[m]} \rightarrow \mathbb{R}_+$ is a coverage function if there exists a universe $U$ of elements with non-negative weights and $m$ sets $A_1, \ldots, A_m \subseteq U$ such that for all $S \subseteq [m]$, $f(S)$ is the total weight of elements in $\cup_{j \in S} A_j$. A coverage function is succinct if $|U|$ is at most a fixed polynomial in $m$.

The complexity of partial function extension has been studied earlier for other function classes, with a number of important applications shown. For boolean functions, Boros et al. present complexity results for extension to a large number of boolean function classes, as well as results on approximate extension~\cite{BorosIM98}. Pitt and Valiant show a direct relation between the complexity of partial function extension problem and proper PAC-learning. Informally, a class $\mathcal{F}$ of (boolean) functions on $2^{[m]}$ is said to be properly PAC-learnable if for any distribution $\mu$ on $2^{[m]}$ and any small enough $\epsilon > 0$, any function $f^* \in \mcF$ can be learned by a polynomial-time algorithm that returns a function $f \in \mcF$ with a polynomial number of samples that differs from $f^*$ with probability at most $\epsilon$.  Pitt and Valiant show that if partial function extension for a class $\mcF$ of functions is NP-hard, then the class $\mcF$ cannot be PAC-learned unless NP = RP~\cite{PittV88}.\footnote{Randomized Polynomial (RP) is the class of problems for which a randomized algorithm runs in polynomial time, always answers correctly if the input is a `no' instance, and answers correctly with probability at least $1/2$ if the input is a `yes' instance.} They show computational lower bounds for various classes of boolean functions, thereby obtaining lower bounds on the complexity for learning these classes. In this paper, we show lower bounds on partial function extension for coverage functions, which by this relation give lower bounds on proper PAC learning as well. In separate work, we present results on the computational complexity of partial function extension for submodular, subadditive, and convex functions, and show further connections with learning and property testing~\cite{bhaskar2018partial}.

Characterizing partial functions extendible to convex functions is widely studied in convex analysis. Here a partial function is given defined on a non-convex set of points, and is required to be extended to a convex function on the convex hull or some other convex domain. Characterizations for extendible partial functions are given in various papers, such as~\cite{DragomirescuI92,Yan12}. This finds many applications, including mechanism design~\cite{FrongilloK14}, decision making under risk~\cite{PetersW86}, and quantum computation~\cite{Uhlmann10}.

Another example of the ubiquity of partial function extension is in property testing. Given oracle access to a function $f$, the goal of property testing is to determine by querying the oracle if the function $f$ lies in some class $\mcF$ of functions of interest, or is far from it, i.e., differs from any function in $\mcF$ at a large number of points. Partial function extension is a natural step in property testing, since at any time the query algorithm has a partial function consisting of the points queried and the values at those points. If at any time the partial function thus obtained is not extendible to a function in $\mcF$, the algorithm should reject, and should accept otherwise. Partial function extension is used to give both upper and lower bounds for property testing~\cite{bhaskar2018partial,submodularity}. Partial function extension is thus a basic problem that finds application in a wide variety of different fields.

\paragraph*{Our Contribution} Our input is a partial function $H = \{(T_1,f_1),\dots,(T_n,f_n)\}$ with $T_i \subseteq [m]$ and $f_i \ge 0$, and  the goal is to determine if there exists a coverage function $f:2^{[m]} \rightarrow \mathbb{R}_{\ge 0} $ such that $f(T_i) = f_i$ for all $i \in [n]$. This is the Coverage Extension problem.
Throughout the paper we use $[m]$ for the ground set, $n$ for the number of defined sets in the partial function, and $\mathcal{D}$ for the set of defined sets $\{T_1,\dots,T_n\}$.  We also use $d = \max_{i \in [n]} |T_i|$ to denote the maximum size of sets in  $\mathcal{D}$, and $F := \sum_{i \in [n]} f_i$.

Our first result shows that Coverage Extension is NP-hard.  Interestingly, we show if there exists a coverage function extending the given partial function  then there is an extension by a coverage function for which the size of the universe $|U|$ is at most $n$. This shows that Coverage Extension is in NP. In contrast, it is known that minimal certificates for non-extendibility may be of exponential size~\cite{coverage}. Also, unlike property testing, this shows that Coverage Extension does not become easier when restricted to succinct coverage functions.

\begin{theorem}
	\label{extension-lb}
	Coverage Extension  is NP-complete.
\end{theorem}

For the hardness, we show a reduction from fractional graph colouring, a problem studied in fractional graph theory. Our hardness for extension also shows the following result for proper learning of succinct coverage functions.

\begin{theorem}
\label{pac-learning}
Unless RP = NP, the class of succinct coverage functions cannot be properly PAC-learned (i.e., cannot be PMAC-learned with approximation factor $\alpha = 1$).
\end{theorem}

These are the first hardness results for learning coverage functions based on standard complexity assumptions. Earlier results showed a reduction from learning disjoint DNF formulas to learning coverage functions~\cite{feldman2014learning}, however as far as we are aware, there are no known lower bounds for learning disjoint DNF formulas. The following theorem is shown in the appendix.

Given the hardness result for Coverage Extension,  we study approximation algorithms for two natural optimization versions of the extension problem. In both of these problems, the  goal is to determine the distance between the given partial function and the class of coverage functions. 
Based on the notion of the distance, we study the following two problems.



In  \emph{Coverage Approximate Extension}, the goal is to determine minimum value of $\alpha \ge 1$ such that there exists a coverage function $f:2^{[m]} \rightarrow \mathbb{R}_{\ge 0} $ satisfying $f_i \le f(T_i) \le \alpha f_i$ for all $i \in [n]$.

In  \emph{Coverage Norm Extension}, the goal is to determine the minimum $L_1$ distance from a coverage function, i.e.,  minimize $\sum_{i \in [n]} |\epsilon_i|$  
where $\epsilon_i = f(T_i) - f_i$ for all $i \in [n]$ for some coverage function $f$.  

The two notions of approximation we study thus roughly correspond to the two prevalent notions of learning real-valued functions. Coverage Approximate Extension corresponds to PMAC learning, where we look for point-wise multiplicative approximations. Coverage Norm Extension corresponds to minimizing the $L_1$ distance in PAC learning.

Throughout this paper, the minimum value of $\alpha$ in Coverage  Approximate Extension will be denoted by $\alpha^*$ and minimum value of $\sum_{i \in [n]} |\epsilon_i|$ in Norm Extension will be denoted by $OPT$. As both of these problems are generalisations of Coverage Extension, they are NP-hard. We give upper and lower bounds for approximation for both of these problem.


\begin{theorem}
	\label{extension-ub}
	There is a $\left( \min \{d,m^{2/3}\} \log d\right)$-approximation algorithm for Coverage Approximate Extension. If $d$ is a constant then there is a $d$-approximation algorithm. 
\end{theorem}

In Coverage Norm Extension,  $OPT = 0$ iff the partial function is extendible and  hence no multiplicative approximation is possible for $OPT$ unless P = NP (because of Theorem \ref{extension-lb}).  We  hence consider additive approximations for Coverage Norm Extension.   An algorithm for Coverage Norm Extension is called an $\alpha$-approximation algorithm if for all instances (partial functions), the value $\beta$ returned by the algorithm  satisfies $OPT \le \beta \le OPT + \alpha$.
We show nearly tight upper and lower bounds on the hardness of approximation. 
As defined before $F = \sum_{i \in [n]} f_i$. Note that an $F$-approximation algorithm is trivial, since the function $f=0$ is coverage and satisfies $\sum_{i \in [n]} |f(T_i) - f_i| \le F$. 
\begin{theorem}
	\label{norm-ub}
	There is a $(1- 1/d) F$-approximation algorithm for Coverage Norm Extension. Moreover,  a coverage function $f$ can be efficiently computed such that $\sum_{i \in [n]} |f(T_i)-f_i| \le OPT + (1- 1/d) F$.
\end{theorem}
\begin{theorem}
	\label{norm-lb}
	It is NP-hard  to approximate Coverage Norm Extension by a  factor $\alpha = 2^{poly(n,m)} F^\delta$ for any fixed $0 \le \delta <1$. This holds  even when $d = 2$.
\end{theorem}

Our lower bound is roughly based on the equivalence of \emph{validity} and \emph{membership}, where given a convex, compact set $K$, the validity problem is to determine the optimal value of $c^Tx$ given a vector $c$ over all $x \in K$, while the membership problem seeks to determine if a given point $x$ is in $K$ or not. The equivalence of optimization and separation is a widely used tool. The reduction from optimization to separation is particularly useful for, e.g., solving linear programs with exponential constraints. Our work is unusual in both the use of validity and membership rather than optimization and separation, and because of the direction --- we use the equivalence to show hardness of the validity problem. We hope that our techniques may be useful in future work as well.

\paragraph*{Related Work} We focus here on work related to partial function extension and coverage functions. In a separate paper, we study partial function extension to submodular, subadditive, and convex functions, showing results on the complexity as well as applications to learning and property testing~\cite{bhaskar2018partial}.
 Previously, Seshadri and Vondrak~\cite{submodularity}  introduce the problem of extending partial functions  to a submodular function, and note its usefulness in analyzing  property testing algorithms. For submodular functions, partial function extension is also useful in optimization \cite{topkis1978minimizing}. The problem of extending a partial function to a convex function is also studied in convex analysis \cite{Yan12, DragomirescuI92}. As mentioned earlier, both characterizing extendible partial functions, and the complexity of partial function extension has been studied for large classes of Boolean functions~\cite{BorosIM98,PittV88}.

Chakrabarty and Huang study   property testing for coverage functions~ \cite{coverage}. Here, the goal is to determine  whether the input function (given by an oracle) is coverage or far from coverage by  querying  an oracle, where distance is measured by the number of points at which the function must be changed for it to be coverage.   They show that succinct coverage functions can be reconstructed with a polynomial number of queries and hence can be efficiently tested. However, they conjecture that testing general coverage functions  requires  $2^{\Omega(m)}$ queries, and prove this lower bound under a different notion of distance. They present a particular characterization of coverage functions in terms of the $W$-transform that we use as well.

There has also been interest in sketching and learning coverage functions. Badanidiyuru et al. \cite{badanidiyuru2012sketching} showed that coverage functions admit a  $(1+ \epsilon)$-sketch, i.e.,  given any coverage function, there exists a succinct coverage function (of size polynomial  in $m$ and $1/\epsilon$) that approximates the original function within $(1+\epsilon)$ factor with high probability. 
 Feldman and Kothari \cite{feldman2014learning} gave a fully polynomial time algorithm for learning succinct coverage functions in the PMAC model if the distribution is uniform. However, if the distribution is unknown, they show learning coverage functions is as hard as learning polynomial size DNF formulas for which no efficient algorithm is known.
  
 Balkanski et al \cite{balkanski2017limitations} study whether coverage functions can be optimized from samples. They consider a scenario where random  samples  $\{(S_i,f(S_i))\}$ of an unknown coverage function $f$ are provided  and ask if it is possible to optimize $f$ under a cardinality constraint, i.e., solve $\max_{S: |S| \le k|}f(S)$. They prove a negative result: no algorithm can achieve approximation ratio better than $2^{\Omega(\sqrt{\log m})}$ with a polynomial number of sampled points.
 
 

		\section{Preliminaries}

As earlier, for $m \in \mathbb{Z}_+$, define $[m] := \{1, 2, \ldots, m\}$. A set function $f$ over a ground set $[m]$ is a coverage function if there exists a universe $U$ of elements with non-negative weights  and $m$ sets $A_1,...,A_m \subseteq U$, such that for all $S \subseteq [m], f(S)$ 
 is the total weight of elements in $\cup_{j \in S} A_j$. A coverage function is \emph{succinct} if $|U|$ is at most a fixed polynomial in $m$. 

  Chakrabarty and Huang \cite{coverage} characterize coverage functions in terms of their $W$-transform, which we use as well. For a set function $f:2^{[m]} \rightarrow \mathbb{R}_{\ge 0} $, the $W$-transform $w:2^{[m]} \setminus \emptyset \rightarrow \mathbb{R}$ is defined as 
  \begin{equation}
  \forall S \in  2^{[m]} \setminus \emptyset, \quad w(S) = \displaystyle\sum_{T: S \cup T = [m]} (-1)^{|S \cap T| +1} f(T)  \, .
  \label{eqn:wcoeff1}
  \end{equation} 	
  The set  $\{w(S)| S \in \Mplus \}$ is called the set of $W$-coefficients of $f$. We can also recover the function $f$ from its $W$-coefficients.
  \begin{equation}
  \forall T \subseteq  [m], \quad f(T) = \displaystyle\sum_{S \subseteq [m]: S \cap T \neq \emptyset} w(S)  \, .
  \label{eqn:wcoeff2}
  \end{equation}

  If $f$ is a coverage function induced by the universe $U$ and sets $A_1,\dots,A_m$, then the $W$-transform $w(S)$ is precisely the weight of the set $\{(\cap_{i \in S} A_i) \setminus \cup_{j \not \in  S} A_j\}$, and is hence non-negative. The converse is also true. The set $\{S| w(S) > 0\}$ is the called the \emph{support} of the coverage function, and the elements are exactly the elements of the universe $U$. 

  \begin{theorem}\cite{coverage}
  	\label{mobius-inversion}
  	A set function $f:2^{[m]} \rightarrow \mathbb{R}_{\ge 0} $ is a coverage function iff all of its $W$-coefficients are non-negative. 
  \end{theorem}

 From Theorem \ref{mobius-inversion}, given a partial function $H$, there exists a coverage function $f$ satisfying $f(T_i) = f_i$ for all $ i \in [n]$ iff  the following linear program   is feasible, where the variables are the $W$-coefficients $w(S)$ for all $S \in  2^{[m]} \setminus \emptyset$. \\

 \[  \text{Extension-P: } \qquad \displaystyle\sum_{S : S \cap T_i \neq \emptyset } w(S) = f_i \quad \forall i \in [n]\, , \qquad  w(S) \ge 0 \quad  \forall S \in  2^{[m]}\setminus \emptyset. \]
 
  All missing proofs are in the appendix.
  
 \section{Coverage Extension and PAC-Learning}
Our first observation is that there is a polynomial-sized certificate of extendibility to a coverage function. This is obtained by observing that at a vertex of the feasible set in Extension-P, at most $n$ of the variables are non-zero. It is interesting to compare this with Chakrabarty and Huang~\cite{coverage}, who give an example to show that minimal certificates of nonextendibility may be of exponential size.

 \begin{proposition}
 	\label{succinct}
 	If a partial function is extendible to a coverage function, then it is also extendible to a coverage function with support size $\le n$. Hence, Coverage Extension is in NP.
 \end{proposition}
We show the  NP-hardness of Coverage Extension by reduction from \emph{fractional chromatic number}, defined as follows. Given a graph $G = (V,E)$, a set $I \subseteq V$ is called an \emph{independent set} if no two vertices in $I$ are adjacent. Let $\mathcal{I}$ be the set of all independent sets.  The fractional chromatic number $\chi^*(G)$ of a graph $G$ is the optimal value of the following  linear program.
\[\chi^*(G) := \left \{\min \sum_{I \in \mathcal{I}} x_I : \sum_{I \in \mathcal{I}: v \in I} x_I \ge 1 \quad \forall v \in  V(G), \, 0 \le x_I \le 1 \quad \forall  I \in \mathcal{I} \right \}\]
Note that if $x_I \in \{0,1\}$ then the optimal value is just the chromatic number of the graph.\footnote{The chromatic number of a graph is the  minimum number of colours required to colour the vertices so that no two adjacent vertices get the same colour.}
\begin{theorem}\cite{godsil2013algebraic}
	\label{frac-colouring}
	For graph $G=(V,E)$, there exist nonnegative weights $\{x_I\}_{I \in \mathcal{I}}$ on independent sets such that $\chi^*(G) = \sum_{I \in \mathcal{I}} x_I$ and $\sum_{I \in \mathcal{I}: v \in I} x_I = 1 \quad \forall v \in  V$.
\end{theorem}
 \begin{corollary}
 	\label{equality-x}
 For graph $G=(V,E)$ and for any value of $t$ such that $\chi^*(G) \le t \le |V|$, there exist nonnegative weights $\{z_I\}_{I \in \mathcal{I}}$ on independent sets such that  $ \sum_{I \in \mathcal{I}} z_I = t$ and $\sum_{I \in \mathcal{I}: v \in I} z_I = 1 \quad \forall v \in  V$.
 \end{corollary}
 \begin{theorem}\cite{khot2001improved}
 	\label{colouringhard}
 	Given graph $G=(V,E)$ and $1 \le k \le |V|$, it is NP-hard to determine if $\chi^*(G) \le k$.
 \end{theorem}

We now show the NP-hardness of Coverage Extension.

 \begin{proof}[Proof of Theorem~\ref{extension-lb}.]
 Since membership in NP was shown earlier, we give the reduction from  fractional chromatic number.  The input is a graph $G=(V,E)$ and  $1 \le k \le |V|$. 
 	
 	We identify $[n']$ with the set of vertices $V$, and therefore $E(G) \subseteq \{\{i,j\}| i,j \in [n']\}$, and any set $S \subseteq [n']$ can be viewed as a set of vertices. The partial function construction is as follows. The ground set is $[n']$ and therefore $m = n'$. The partial function is defined at all vertices, all edges, and the set consisting of all vertices. Hence $\mathcal{D} = \{\{i\}| i \in [n']\} \cup E(G) \cup \{[n']\}$ and $|\mcD| = n' + |E(G)| + 1$. 
 	 The value of the partial function $h$ at these defined sets is given by 
 		 \[   
 		 h(S) = 
 		 \begin{cases}
 		 1 &\quad\text{if} \quad S  = \{i\}, \, i \in [n']\, ,\\
 		 2 &\quad\text{if} \quad S \in  E(G) \, ,\\
 		 k &\quad\text{if } \quad S = \{[n']\}\, .
 		 \end{cases}
 		 \]
 We claim that $\chi^*(G) \le k $ iff the above partial function is extendible. Suppose $\chi^*(G) \le k $. Therefore by Corollary \ref{equality-x}, there exist nonnegative weights $\{x_I\}_{I \in \mathcal{I}}$ such that  $ \sum_{I \in \mathcal{I}} x_I = k$ and $\sum_{I \in \mathcal{I}: v \in I} x_I = 1 \quad \forall v \in  V(G)$. For all $S \in 2^{[m]} \setminus \emptyset$, define the function $w(S)$ as $x_S$ if $S \in \mathcal{I}$ and $0$ otherwise. Since $w(S) \ge 0$, this defines the $W$-transform for a coverage function $g$. We have, for any $i \in [n']$, 
 
 \[
 g(\{i\}) = \sum_{S : S \cap \{i\} \neq \emptyset } w(S) = \sum_{I \in \mathcal{I}: i \in I} x_I = 1 \, ,
 \]
 
 \noindent for any $\{i,j\} \in E(G)$, 
 
 \[
 g(\{i,j\}) = \sum_{S : S \cap \{i,j\} \neq \emptyset } w(S) = \sum_{I \in \mathcal{I}: i \in I} x_I +  \sum_{I \in \mathcal{I}: j \in I} x_I = 2
 \]
 
 \noindent as no independent set $I$ can  contain both $i$ and $j$; and finally $g(\{[n']\}) = \sum_{S : S \cap \{[n']\} \neq \emptyset } w(S) = \sum_{I \in \mathcal{I}} x_I = k$. Therefore $g$ is an extension of the above partial function $h$.
 
 Now suppose there is an extension, i.e., there exists $w(S) \ge 0$ for all $S \in \Mplus$ such that for any $i \in [n']$, $\sum_{S : S \cap \{i\} \neq \emptyset } w(S)  = 1$; for any $\{i,j\} \in E(G),  \sum_{S : S \cap \{i,j\} \neq \emptyset } w(S)  = 2$; and finally $ \sum_{S : S \cap \{[n']\} \neq \emptyset } w(S) = k$. For any $\{i,j\} \in E(G)$, we have 
 
 \[
 \sum_{S : S \cap \{i,j\} \neq \emptyset } w(S)  = \sum_{S : S \cap \{i\} \neq \emptyset } w(S)  + \sum_{S : S \cap \{j\} \neq \emptyset } w(S) - \sum_{S : S \supseteq \{i,j\} } w(S) \, .
 \]
 
 \noindent Therefore,$\sum_{S : S \supseteq \{i,j\} } w(S) = 0$, i.e.,  if $w(S) > 0$ then $S$ must be an independent set.  It now follows that $\chi^*(G) \le \sum_{S : S \cap \{[n']\} \neq \emptyset } w(S) = k$.
 \end{proof}

  \subsection*{Proper PAC-learning of Coverage functions}

We now prove Theorem~\ref{pac-learning}. We first recall the definition of PAC-learning.

  \begin{definition}[\cite{BalcanH11}]
  	An algorithm $\mathcal{A}$ properly PAC-learns a family of functions $\mathcal{F} $, if for \emph{any} distribution $\mu$ (on $2^{[m]}$) and \emph{any} target function $f^* \in \mathcal{F}$,  and  for \emph{any} sufficiently small $\epsilon,\delta >0$:
  	\begin{enumerate}
  		\item $\mathcal{A}$ takes the sequence  $\{(S_i,f^*(S_i))\}_{1 \le i \le l}$  as input where  $l$ is $poly(m,1/\delta,1/\epsilon)$ and the sequence $\{S_i\}_{1 \le i \le l}$ is drawn i.i.d. from the distribution $\mu$, 
  		\item $\mathcal{A}$  runs in $poly(m,1/\delta,1/\epsilon)$ time, and
  		\item $\mathcal{A}$  returns a function $f:2^{[m]} \rightarrow \mathbb{R} \in \mcF$ such that 
  		\[
  		Pr_{S_1,\dots,S_l \sim \mu} \big[Pr_{S \sim \mu} [  f(S) = f^*(S)] \ge 1 - \epsilon \big] \ge 1 - \delta
  		\]
  	\end{enumerate}
  \end{definition}

We use the reconstruction algorithm for coverage functions given by Chakrabarty and Huang~\cite{coverage} in our proof. Given a coverage function $f$ as an input, this reconstruction algorithm terminates in $O(m \, s)$ steps where $s$ is the support size of $f$, i.e., the number of non-zero $W$-coefficients of $f$, and returns these non-zero $W$-coefficients. 

  Recall the reduction from fractional chromatic number  to Coverage Extension (Theorem \ref{extension-lb}). 
  Given an instance of fractional chromatic number (graph $G=(V,E)$ and rational $k'$ with $|V|=n'$),  the instance of Coverage Extension is a set of defined points $\mcD = \{\{i\} | i \in [n']\} \cup E(G)  \cup \{[n']\}$ and a function $h$ on $\mcD$. Let $k = |\mcD| = |V| + |E| +1$. From Theorem \ref{extension-lb} and Proposition~\ref{succinct},  $\chi^*(G) \le k'$ iff $h$ is extendible to a coverage function with support size at most $k$.

   Let $\mathcal{F} $ be a family of coverage functions with support size at most $k$. Let $\epsilon = 1/k^3$ (and hence $\epsilon < 1/|\mcD|$) and $\mu$ be a uniform distribution  over   $\{(S,h(S))| S \in \mcD\}$. Now suppose a (randomized) algorithm $A$ properly PAC-learns $\mathcal{F}$. We will show that in this case, we can determine efficiently if the partial function is extendible to a coverage function, and hence RP = NP.
  
  Suppose the  algorithm $A$ returns a function $g$.
  If the partial function  is extendible then  there exists a function in $\mathcal{F}$ that has the same value on samples seen by $A$. Therefore, if the partial function is extendible  then  $g(S)$ must be equal to $h(S)$ for all $S \in \mcD$ (since $\epsilon < 1/|\mcD|$ and $A$ must satisfy$ Pr_{S \sim D^*} [  f(S) = f^*(S)] \ge 1 - \epsilon$).
   We run the reconstruction algorithm on input $g$. If the partial function is extendible then $g$ must be in $\mathcal{F}$ and hence the reconstruction algorithm must terminate in $O(mk)$ steps. Further, if  $\{w(S)\}_{S \in \mathcal{S}}$ is the output of the algorithm then
  (i) $w(S) > 0 $ for all $S \in \mathcal{S}$, (ii) $|\mathcal{S}| \le k$  (iii) the coverage function $f'$ given by the $W$-coefficients $w'(S) = w(S)$ if $S \in \mathcal{S}$ and $0$ otherwise is an extension of the partial function $h$. Condition (iii) should hold because $f'$ must be the same as $g$ which we have shown earlier is an extension of $h$. 
  
  The converse is also true --- if the reconstruction algorithm terminates and (i), (ii), (iii) hold then clearly $h$ is extendible (by $f'$). Since all the steps require polynomial time to check, we can efficiently determine if the partial function is extendible.

 \section{Coverage Approximate Extension} 
 
 We now build the framework for Theorem~\ref{extension-ub}.
 We start with the following lemma.
 \begin{lemma}
 	\label{coverage-claim}
 	Given a partial function $H$ and $\alpha \ge 1$, there is no coverage function $f$ satisfying $f_i \le f(T_i) \le \alpha f_i$ for all $i \in [n]$  iff the following program, with variables $l_i$ for all $i \in [n]$ is feasible: 
 		\begin{equation}
 	\label{inequality-cov-2}
 	-\alpha \sum_{i : l_i < 0 } f_i l_i < \sum_{i : l_i > 0 } f_i l_i 
 	\end{equation}
 	\begin{equation}
 	\label{inequality-cov-1}
 	\sum_{i : S \cap T_i \neq \emptyset} l_i \le 0 \quad  \forall S \subseteq [m]
 	\end{equation}
 \end{lemma}
 Thus the optimal approximation ratio $\alpha^*$ is the minimum value of $\alpha$ for which~\eqref{inequality-cov-2} and ~\eqref{inequality-cov-1} are not feasible together.
 
A natural representation of the partial function $H = \{(T_1,f_1),\dots,(T_n,f_n)\}$ is as a weighted bipartite graph  $H = (A \cup [m], E)$ with  $|A| = n$,  and an edge between $a_i \in A$ and $j \in [m]$ if the set $T_i$ contains element $j \in [m]$. Each vertex $a_i \in A$ also has weight $f_i$. Then $d = \max_{i} |T_i|$ is the maximum degree of any vertex in $A$. For the remainder of this section, we will use this representation of partial functions.
 
We use the following notation given a bipartite graph $H = (A \cup [m],E)$.   For any $S \subseteq [m]$,  let $N(S) = \{v \in A: (v,j) \in E \quad \text{for some} \quad  j \in  S\}$ be the set of neighbours of set $S$. Similarly for set $R \subseteq A$, $N(R) = \{j \in [m]: (v,j) \in E \quad \text{for some} \quad  v \in  R\}$ be the set of neighbours of set $R$. For any vertex $v$ in $H$, we use $N(v)$ for $N(\{v\})$.  In this  bipartite graph representation, the inequality (\ref{inequality-cov-1}) is equivalent to $\sum_{i \in N(S) } l_i \le 0$ for all $ S \subseteq [m]$. 

 We now define a parameter $\kappa$ called the \emph{replacement ratio} for a partial function $H$.
 
   \begin{definition}
   Let $H = (A \cup [m],E)$ be a bipartite graph with weights $f_v$ on each $v \in A$. For $v \in A$, let $\mathcal{F}_v = \{R \subseteq A \setminus \{v\} \, | \, N(R) \supseteq N(v)\}$ be the set of all subsets of $A \setminus \{v\}$ that cover all the neighbours of $v$. We call each $R \in \mathcal{F}_v$ a \emph{replacement} for $v$. The \emph{replacement ratio} $\kappa$ is then the minimum of $\frac{\sum_{w \in R} f_w}{f_v}$ over all vertices $v \in A$ and replacements $R \in \mathcal{F}_v$.
   \end{definition}
  The proof of the upper bound in Theorem~\ref{extension-ub} will follow from the bounds on $\alpha^*$ shown in Lemma \ref{lower-bound-alpha*}, \ref{upper-bound-alpha*-d} and \ref{upper-bound-alpha*-m}.
  \begin{lemma}
  	\label{lower-bound-alpha*}
  	For any partial function $H$, 
  	$ \alpha^* \ge \frac{1}{\kappa} $.
  \end{lemma}
  \begin{proof}
   	 By definition of $\kappa$, there exists a vertex $v \in A$ and a replacement $R$ for $v$ such that $\sum_{w \in R} f_w = c f_v$. Note that setting $l_w = -1 \quad \forall w \in R$, $l_v = 1$ and all other $l_w$'s to be zero results in feasibility of the inequalities $  \sum_{w \in N(S) } l_w \le 0 $ for all $S \in \Mplus$. From the definition of $ \alpha^*$ and Lemma~\ref{coverage-claim}, $\alpha^* \sum_{w\in R}f_w \ge  f_v$, and hence $ \alpha^* \ge 1/\kappa$.
  \end{proof}
  
  	Let $\beta = \frac{\min \{d,m^{2/3}\}}{\kappa}$. Given values $\{l_v\}_{v \in A}$ on the vertices in $A$ such that $\sum_{ v \in N(S)} l_v \le 0$ for all $S \subseteq [m]$, we will show that $\beta \sum_{v : l_v < 0 } f_v l_v \ge \sum_{v : l_v > 0 } f_v l_v$ and hence $\alpha^* \le \beta$. If $l_v = 0$ for any vertex, we simply ignore such a vertex, since it does not affect either~\eqref{inequality-cov-1} or~\eqref{inequality-cov-2}.

  	 By scaling, we can assume that $l_v \in \mathbb{Z}$ for all $v \in A$. At some point, we will use Hall's theorem to show a  perfect matching. To simplify exposition, we  replace each $v \in A$ with $|l_v|$ identical copies, each of which is adjacent to the same vertices as $v$. Each such new vertex $v'$ has $l_{v'} = 1$ if $l_v > 0$ and $l_{v'} = -1 $ if $l_v < 0$, and $f_{v'} = f_v$. Let the new bipartite graph be $H' = (A' \cup [m],E')$. It is easy to check that in the new bipartite graph, the degree of vertices in $A'$ and the values $\kappa$, $\sum_{v \in A' : l_v > 0 } f_v l_v$, $\sum_{v \in A' : l_v < 0 } f_v l_v $ and $\sum_{v \in N(S)} l_v$ remain unchanged for all $S \subseteq [m]$. 

  	Let $\mathcal{N} = \{v \in A' | l_v = -1\}$ and $\mathcal{P} = \{v \in A' | l_v =1\}$, and let $E^-$ be the set of edges with one end-point in $\mcN$, while $E^+$ are the edges with one end-point in $\mcP$. 
  	For any $S \subseteq [m]$, let $N^-(S) = N(S) \cap \mathcal{N}$ and $N^+(S) = N(S) \cap \mathcal{P}$ (so $N(S) = N^+(S) \cup N^-(S)$). Finally, define $E^+(S)$ $(E^-(S))$ as the set of edges with one end-point in $S$ and the other end-point in $\mcP$ $(\mcP)$. If $S = \{j\}$, we abuse notation slightly and use $N^-(j)$, $N^+(j)$, $E^-(j)$ and $E^+(j)$. Note that $|N^-(S)| \ge |N^+(S)|$ for all $S \subseteq [m]$  in $H'$, since in $H$, $\sum_{ v \in N(S)} l_v \le 0$ for all $S \subseteq [m]$. Our goal is to show  $\beta \sum_{v \in \mcN } f_v  \ge \sum_{v \in \mcP } f_v$.
  	
  	\begin{lemma}
  	\label{hall-theorem}
  	 Suppose for some $\beta' \ge 1$, $\beta' |N^-(S)| \ge  \sum_{ j \in S} |N^+(j)|$ for all $S \subseteq [m]$. Then for each vertex $v \in \mcP$, there exists a replacement $F_v \subseteq \mcN$ such that each vertex in $\mcN$ is contained in $F_v$ for at most $\beta'$ vertices $v \in \mcP$. Hence, $\beta' \sum_{v \in \mcN} f_v \ge \kappa \sum_{v \in \mcP} f_v$ and so $\alpha^* \le \frac{\beta'}{\kappa}$.
  	\end{lemma}
  	\begin{proof}
  	  By Hall's theorem, there exists a set of edges $M \subseteq E^-$ such that (i)  the degree in $M$ of each vertex $j \in [m]$ is at least $|N^+(j)|$, and (ii) the degree in $M$ of each vertex $v \in \mcN$ is at most $\beta'$. Because of (i), for each $j \in [m]$ there is an injection $h_j$ from edges in $E^+(j)$ to edges in $E^-(j) \cap M$, i.e., each edge in $E^+(j)$ maps to a distinct edge in $E^-(j) \cap M$. Now for a vertex $v \in \mcP$, consider a neighbouring vertex $j \in N(v)$. Each such edge $(v,j)$ is in $E^+(j)$, and is hence mapped by $h_j$ to an edge in $E^-(j) \cap M$. Let $F_v$ be the end-points in $\mcN$ of these mapped edges. That is, $w \in F_v$ iff there exists $j \in N(v)$ such that $(w,j) = h_j(v,j)$. Then $F_v$ is a replacement for $v$, and hence, $\sum_{w \in F_v} f_w \ge \kappa f_v$. Further, because of (ii), and since each $h_j$ is an injection, each vertex in $\mcN$ appears in $F_v$ for at most $\beta'$ vertices $v \in \mcP$. Then summing the inequality $\sum_{w \in F_v} f_w \ge \kappa f_v$ over all $v \in \mcP$, we get that $\beta' \sum_{v \in \mcN} f_v \ge \kappa \sum_{v \in \mcP} f_v$ as required.
  	\end{proof}
  	 \begin{lemma}
  	\label{upper-bound-alpha*-d}
  	For any partial function $H$, 
  	$\alpha^* \le \frac{d}{\kappa}$.
  \end{lemma}
  \begin{proof}
   Fix $S \subseteq [m]$. Since $|N^-(j)| \ge |N^+(j)|$ for all $j \in [m]$, $\sum_{j \in S} |N^-(j)| \ge \sum_{j \in S} |N^+(j)|$, and since $d$ is the maximum degree of any vertex in $A'$, $d |N^-(S)| \ge \sum_{j \in S} |N^-(j)|$. The proof follows from Lemma~\ref{hall-theorem}.
  \end{proof}
   If we can show $m^{2/3} |N^-(S)| \ge  \sum_{ j \in S} |N^+(j)|$ for all $S \subseteq [m]$ then by Lemma~\ref{hall-theorem}, 	$\alpha^* \le \frac{m^{2/3}}{\kappa}$. Unfortunately this may not be true. Let $\mcN = \{v_1\}, \mcP = \{v_2\}, E^- = \{(v_1,j) | j \in [m]\}$ and $E^+ = \{(v_2,j)|j \in [m]\}$. Note that $ \sum_{ j \in [m]} |N^+(j)| = m$ whereas $ |N^-([m])| =  1$. Notice that in this bad example, the bipartite graph contains a $4$-cycle $v_1,j_1,v_2,j_2,v_1$ where $v_1 \in \mcN$ and $v_2 \in \mcP$. We now define a subgraph called a $\textit{diamond}$ which generalises such a $4$-cycle. A diamond $(v_p,v_n,J)$ of size $k$ is a subgraph of $H'$ where $v_p \in \mcP, v_n \in \mcN, J \subseteq [m]$ ($|J| = k)$ such that for all $j \in J$, both $(v_p,j)$ and $(v_n,j)$ are contained in  $E'$. Note that a $4$-cycle is a diamond of size two (and the bad example considered above is a diamond of size $m$). 
   Let $k_{max} = m^a$ ($0 \le a \le 1$) be the maximum size of any diamond in $H'$. 
  \begin{lemma}
  \label{no-4-cycle}
  For all $S \subseteq [m]$,  $m^{\frac{1+a}{2}} |N^-(S)| \ge \sum_{ j \in S} |N^+(j)|$, where $m^a$ is the size of the largest diamond in $H'$.
  \end{lemma}
  \begin{proof}
   Recall that for all $j \in [m]$, $|N^+(j)| \le |N^-(j)|$, hence there is an injection $h_j$ from $N^+(j)$ to $N^-(j)$, i.e, $h_j$ maps each vertex in $N^+(j)$ to a unique vertex in $N^-(j)$. Fix $S \subseteq [m]$ and vertex $v \in \mcP$, and let $S_v := N(v) \cap S$ be the neighbourhood of $v$ in $S$. We will consider $N^+(S_v)$ and $N^-(S_v)$, the negative and positive neighbourhoods of $S_v$. Note that since all vertices in $S_v$ are adjacent to $v \in \mcP$, a vertex in $N^-(S_v)$ is adjacent to at most $m^a$ vertices in $S_v$, by definition of $a$. Thus for a vertex $v' \in N^-(S_v)$, there are at most $m^a$ different vertices $j \in S_v$ for which $h_j$ maps a vertex in $N^+(j)$ to $v'$, and hence $m^a |N^-(S_v)| \ge \sum_{j \in S_v} |N^+(j)|$.
   
	
Now  if there is a vertex  $v  \in \mcP$ such that $m^{\frac{1-a}{2}}\sum_{j \in S_v} |N^+(j)|$ $\ge \sum_{ j \in S} |N^+(j)| $ then we are done, since
\[ 
|N^-(S)| \ge |N^-(S_v)| \ge \frac{\sum_{j \in S_v} |N^+(j)|}{m^a} \ge \frac{\sum_{ j \in S} d^+_j}{m^{\frac{1+a}{2}}} \,.
\]

\noindent So assume that for all $v \in  \mcP$, $\sum_{j \in S_v} |N^+(j)| \le \frac{\sum_{ j \in S} |N^+(j)|}{m^{\frac{1-a}{2}}}$. In this case, note that by reversing the order of summation,
\[
\sum_{ j \in S} |N^+(j)|^2 = \sum_{j \in S} \sum_{v \in N^+(j)} |N^+(j)| = \sum_{v \in N^+(S)} \sum_{j \in S_v} |N^+(j)| \le |N^+(S)|\,\frac{\sum_{ j \in S} |N^+(j)|}{m^{\frac{1-a}{2}}} \, .
\]

\noindent Therefore, using the above inequality for $|N^+(S)|$, 
\[
|N^-(S)| \ge |N^+(S)| \ge m^{\frac{1-a}{2}} \frac{\sum_{ j \in S} |N^+(j)|^2 }{\sum_{ j \in S} |N^+(j)|} \ge \frac{m^{\frac{1-a}{2}}}{|S|} \frac{\left(\sum_{ j \in S} |N^+(j)\right)^2 }{\sum_{ j \in S} |N^+(j)|} \ge \frac{\sum_{ j \in S} |N^+(j)| }{m^{\frac{1+a}{2}}} \,
\]

\noindent as required by the lemma. The third inequality follows from Cauchy-Schwarz.
\end{proof}
  
From Lemmas~\ref{hall-theorem} and~\ref{no-4-cycle},  if $a \le 1/3$ then $\alpha^* \le \frac{m^{2/3}}{\kappa}$. Next we show this is true in general. 

 \begin{lemma}
  	\label{upper-bound-alpha*-m}
  	For any partial function $H$, 
  	$ \alpha^* \le \frac{m^{2/3}}{\kappa}$.
  \end{lemma}
  \begin{proof}
   If $k_{max} \le m^{1/3}$ then by Lemma \ref{no-4-cycle} and \ref{hall-theorem}, $\alpha^* \le \frac{m^{2/3}}{\kappa}$. So we assume $k_{max} > m^{1/3}$.
        In this case, we pick a diamond $(v_p,v_n,J)$ of  size $> m^{1/3}$. We remove, for all $j \in J$, the edges $(v_p,j)$ and $(v_n,j)$. We repeat the above procedure (in the new graph) until we are left with a bipartite graph where all diamonds are of size at most $m^{1/3}$. Note that if a diamond $(v_p,v_n,J)$ of size $k$ is removed then the degree of $v_n$ decreases by $k$. Hence, for a fixed vertex $v_n$, number of removed diamonds is at most $m^{2/3}$ (as at any step we remove diamonds of size at least $m^{1/3}$). It is easy to see that after every step, $|N^-(S)| \ge |N^+(S)|$ (for all $S \in \Mplus$) still holds in the bipartite graph. Let $H^*$ be the bipartite graph at the end (all diamonds of size at most $m^{1/3}$). 
  
Note that we do not remove any vertex in the above procedure. Fix vertex $v \in \mcP$. By Lemmas~\ref{no-4-cycle} and~\ref{hall-theorem} with $a=1/3$, there exists $F_v \subseteq \mcN$ such that $F_v$ covers all neighbours of $v$ in $H^*$ and   each vertex in $\mcN$ appears in $F_v$ for at most $m^{2/3}$ vertices $v \in \mcP$. Since we have removed edges, $F_v$ may not cover all the neighbours of $v$ in $H'$. Let $v^1,\dots,v^s \in \mcN$ be the set of all vertices such that for each $i \in [s]$, a diamond $(v,v^i,J^i)$ was removed in a removal step. Clearly $\{v^1,\dots,v^s\} \cup F_v $ cover all the neighbour of $v$ in $H'$. Therefore, we have $\sum_{i =1}^{s} f_{v^i} + \sum_{w \in F_v} f_w \ge \kappa f_v$. Since any $v^i$ ($1 \le i \le s$) is a part of at most $m^{2/3}$ removed diamonds and each vertex in $\mcN$ appears in $F_v$ for at most $m^{2/3}$ vertices $v \in \mcP$ , summing the above inequality for each $v \in \mcP$, we get $m^{2/3}\sum_{v \in \mcN} f_v \ge \kappa \sum_{v \in \mcP} f_v$ as required. 
  \end{proof}

It follows from Lemmas~\ref{lower-bound-alpha*},~\ref{upper-bound-alpha*-d} and~\ref{upper-bound-alpha*-m} that an algorithm that returns $\frac{\min \{d,m^{2/3}\}}{\kappa}$ is a $\min \{d,m^{2/3}\}$-approximation algorithm. However, computing $\kappa$ corresponds to solving a general set cover instance, and is \nphard. This connection however allows us to show the following result.

  \begin{lemma}
  		\label{computing-c}
  	Given a partial function, the replacement ratio $\kappa$  can be efficiently approximated by $\kappa'$ such that $\kappa \le \kappa' \le \kappa \log d$. If $d$ is a constant, the replacement ratio  $\kappa$ can be determined efficiently. 
  \end{lemma}
This completes the proof of the upper bound in Theorem~\ref{extension-ub}. We now show that the analysis of our algorithm cannot be substantially improved.

  \begin{theorem}
  	\label{tight-logm-factor}
  	There is a partial function with $\alpha^* = 1/\kappa$ and a partial function with $d = \sqrt{m}$ and $\alpha^* = \Omega(\frac{\sqrt{m}}{\kappa \log m})$.
  \end{theorem}
  \begin{proof}
  	Fix $r \ge 1$, and consider the partial function $H = \{(\{1\},r),(\{2\},r),(\{1,2\},1)\}$ over ground set $\{1,2\}$. Clearly $\kappa = 1/r$, and since any coverage function must be monotone~\eqref{eqn:wcoeff2}, $\alpha^* = r$. 
  	
  	Now we show existence of a partial function with $d = \sqrt{m}, \kappa = 1$ and $\alpha^* \ge \Omega(\frac{\sqrt{m}}{ \log m})$. Let the ground set be $[m]$ and defined sets be $\mathcal{D} = \mathcal{P} \cup \mathcal{N}$. 
  	 Let $\mathcal{P} = \{\{1,\dots,\sqrt{m}\},\{\sqrt{m}+1,\dots,2\sqrt{m}\},\dots,\{m-\sqrt{m}+1,m\}\}$. Thus $|\mathcal{P}| = \sqrt{m}$ and each set in $\mathcal{P}$ has size $\sqrt{m}$.  We set $|\mathcal{N}| = k \sqrt{m} \log m$ for some large constant $k$. Each set in $\mathcal{N}$ is constructed by randomly picking exactly one element from the sets $ \{1,\dots,\sqrt{m}\},\{\sqrt{m}+1,\dots,2\sqrt{m}\},\dots,\{m-\sqrt{m}+1,m\}$. 
  	 Thus the size of each set in $\mathcal{N}$ is also $\sqrt{m}$. The value of the partial function at each set in $\mathcal{N}$ is set to $1$, and at each set in $\mcP$ is $\sqrt{m}$.   
  	 
  	 Recall the bipartite graph description of partial functions. In the appendix, we show that 
  	 \begin{align} 
  	 \label{prob}
  	 Pr\left[ |N(S) \cap  \mathcal{N}| \ge |N(S) \cap  \mathcal{P}| \quad  \forall S \subseteq [m]\right] > 0 \, .
  	 \end{align}
  	  Thus $|N^-(S)| \ge |N^+(S)|$ for all $S \subseteq [m]$.
  	  It can be seen that $\kappa = 1$, since for any $v \in \mathcal{P}$ every vertex in $N(v)$ has a set in $\mcN$ containing it, while for $v \in \mcN$, all sets in $\mcP$ are required to cover it. Thus, the minimum value of $\frac{\sum_{w \in R} f_w}{f_v}$ over all  $R \in \mathcal{F}_v$ is $1$ for $v \in \mcP$, while for any $v \in \mathcal{N}$, the minimum value is $m$. Now we show  $\alpha^* \ge \frac{\sqrt{m}}{ k \log m}$.
  	   Let $l_v = 1$ for all $v \in \mathcal{P}$ and $l_v = -1$ for all $v \in \mathcal{N}$. So   $\sum_{v \in N(S)} l_v \le 0$ for all $S$. Since $\sum_{l_v > 0} f_v l_v = m$ and $\sum_{l_v < 0} f_v l_v = -k \sqrt{m} \log m$, we have $\alpha^* \ge \frac{\sqrt{m}}{ k \log m}$. 
\end{proof}  
\begin{remark}
In the above lemma, the bipartite graph shown to exist  does not have $4$-cycle. For such partial functions, we have $\alpha^* \le \sqrt{m}/\kappa$ (Lemma \ref{hall-theorem} and \ref{no-4-cycle}). Therefore, Theorem \ref{tight-logm-factor} shows that the inequality $ \alpha^* \le \frac{\min \{d,\sqrt{m}\}}{\kappa}$ cannot be improved by more than a $\log m$ factor (for such partial functions). It is an interesting open problem to close the gap between our upper bound of $m^{2/3}/\kappa$ and lower bound of $\sqrt{m}/\kappa \log m$ for general partial functions.
\end{remark}
  	  

 \section{Coverage Norm Extension}
 
 From Theorem \ref{mobius-inversion}, the Norm Extension problem can be stated as the convex program Norm-P. 
  It can be equivalently transformed to a linear program whose dual is Norm-D. 
 
 
  
 \begin{minipage}[t]{0.4\textwidth}
 	Norm-P: \qquad \qquad $\min \sum_{i =1}^{n} |\epsilon_i|$
 	\[ \displaystyle\sum_{S : S \cap T_i \neq \emptyset } w(S)  = f_i + \epsilon_i  \quad \forall i \in [n]\] 
 	\[ w(S) \ge 0 \quad  \forall S \in  2^{[m]}\setminus \emptyset\]
 \end{minipage}
 \hfill \vline \hfill
 \begin{minipage}[t]{0.5\textwidth}
 	Norm-D: \qquad $ \max \quad \displaystyle\sum_{i =1}^{n}f_i y_i $ 
 	\begin{align}
 	\sum_{i : S \cap T_i \neq \emptyset }y_i & \le 0 \quad  \forall S \in  2^{[m]}\setminus \emptyset  \label{giced1} \\
 	-1 \le y_i& \le 1  \quad \forall i \in [n]\label{giced2} 
 	\end{align}
 \end{minipage}

Both Norm-P and Norm-D are clearly feasible. We use $OPT$ for the optimal value of Norm-P (and Norm-D). 
 As stated earlier, no multiplicative approximation is possible for $OPT$ unless P = NP. Therefore, we consider  additive approximations for Norm Extension.   

An algorithm for Norm Extension is called an $\alpha$-approximation algorithm if for all instances (partial functions), the value $\beta$ returned by the algorithm  satisfies $OPT \le \beta \le OPT + \alpha$.
First we prove our upper bound in Theorem~\ref{norm-ub}.  Recall that $d = \max_{i \in [n]} |T_i|$ and $F = \sum_{i \in [n]} f_i$. As noted earlier, the function $f(\cdot)=0$ is trivially an $F$-approximation algorithm for Norm Extension, since $\sum_{i \in [n]} |f(T_i) - f_i| = F$.

\begin{proof}[Proof of Theorem~\ref{norm-ub}.]
 Consider the linear programs obtained by restricting Norm-P to variables $w(S)$ for $S \in [m]$, and similarly restricting the constraints~\eqref{giced1} in Norm-D to sets $S \in [m]$ only. They are clearly the primal and dual of each other. The optimal values of these modified problems (say $OPT^R$, $w^R$ and $y^R$) can be computed in polynomial time. We will show that $OPT \le OPT^R \le OPT + (1-1/d)F$ for the proof of the theorem. The first inequality is obvious, since $OPT^R$ is the optimal solution to a relaxed (dual) linear program.
 
For the second inequality, define $y^A = (y^A_1,\dots,y^A_n)$ as the vector such that for all $i \in [n]$, $y^A_i = y^R_i$ if  $y^R_i \le 0$ and $y^R_i/d$ otherwise. Then note that 

\begin{equation}
OPT^R ~=~ \sum_{i \in [n]} f_i y_i^R ~=~ \sum_{i \in [n]} f_i y_i^A + (1-1/d) \sum_{i: y_i^R \ge 0} f_i y_i^R \le \sum_{i \in [n]} f_i y_i^A + (1-1/d) F \, ,
\label{eqn:fiyi}
\end{equation}

\noindent where the last inequality is because each $y_i^R \le 1$. We now show that $y^A$ is a feasible solution for Norm-D, and hence $\sum_{i \in [n]} f_i y_i^A \le OPT$. Together with~\eqref{eqn:fiyi} this completes the proof.

Clearly $y^A$ satisfies the constraints (\ref{giced2}). We will show that $y^A$ also satisfies the constraints (\ref{giced1}) for all $S \in 2^{[m]}\setminus \emptyset$. Consider any $S \in 2^{[m]}\setminus \emptyset$. Let $P = \{i \in [n]| S \cap T_i \neq \emptyset \quad \text{and} \quad y^R_i > 0\}$ and $N = \{i \in [n]| S \cap T_i \neq \emptyset \quad \text{and} \quad y^R_i \le 0\}$. Thus $P \cup N$ are all sets in $\mcD$ that have nonempty intersection with $S$.  We have for any $j \in S$ that $  	\sum_{i : j \in T_i} \, y^R_i \le 0$. Summing these inequalities over $j \in S$, we obtain $\sum_{i \in P \cup N} |T_i \cap S| y^R_i \le 0$. Thus $\sum_{i \in P }  y^R_i + d \sum_{i \in N}  y^R_i \le 0$. From the definition of $y^A_i$,  we get $\sum_{i:S \cap T_i \neq \emptyset }y^A_i  \le 0$, as required.
\end{proof}


We now prove the lower bound in Theorem~\ref{norm-ub}. We start with an outline of the proof. In a nutshell, the proof shows the following reductions (for brevity, WM stands for Weak Membership and WV for Weak Validity):

\[
\text{Densest-Cut} \le_p \text{Cut WM} \le_p \text{Span WM} \equiv \text{Coverage WM} \le_p \text{Coverage WV} \le_p \text{Norm Extension} \, .
\]

Given a graph $G=(V,E)$ and a positive rational $M$, the \emph{Densest-Cut} problem asks if there is a cut $S \subset V$ such that $\frac{|\delta(S)|}{|S|\,|V \setminus S|} > M$. The Densest-Cut problem is known to be NP-hard~\cite{densest-cut}, and ultimately we reduce the Densest-Cut problem to the problem of  approximating the optimal value for Norm-P. We formally define the other problems later. However, to show this reduction, we need to utilize the equivalence of optimization (or validity) over a polytope and membership in the polytope. Typically optimization algorithms use the equivalence of optimization and separation to show upper bounds, e.g., that a linear program with an exponential number of constraints can be optimized. Our work is unique in that we use the less-utilized equivalence of validity and membership; and secondly, we use it to show hardness. In fact, since we are looking for hardness of approximation algorithms, our work is complicated further by the need to use \emph{weak} versions of this equivalence.

Given a convex and compact set $K$ and a vector $c$, the \emph{Strong Validity} problem, given a vector $c$, is to find the maximum value of $c^T x$ such that $x \in K$ (the $x$ which obtains this maximum is not required). In the \emph{Strong Membership} problem, the goal is to determine if a  given vector $y$ is  in $K$ or not. The \emph{Weak Validity} and \emph{Weak Membership} problems are weaker versions of the Strong Validity and Strong Membership problems respectively, formally defined later. Then Theorem 4.4.4 in~\cite{ellipsoid} says that for a convex and compact body $K$, there is an oracle polynomial time reduction from the Weak Membership problem for $K$ to the Weak Validity problem for $K$.


To formally state Theorem 4.4.4 from~\cite{ellipsoid}, which will form the basis of our reduction, we need the following notations and definitions.
 
 We use $||.||$ for the Euclidean norm. Let $K \subseteq \mathbb{R}^{n'}$ be a convex and compact set. A ball of radius $\epsilon > 0$ around $K$ is defined as 
  \[S(K,\epsilon) := \{x \in \mathbb{R}^{n'} | \thinspace  ||x -y|| \le \epsilon \thinspace \thinspace \text{for some $y$ in $K$}\} \, .\] 
  Thus, for $x \in \mathbb{R}^{n'}$, $S(x,\epsilon)$ is the ball of radius $\epsilon$ around $x$. The interior $\epsilon$-ball of $K$ is defined as 
  \[S(K,-\epsilon) := \{x \in K | S(x,\epsilon) \subseteq K\}\]
  Thus $S(K,-\epsilon)$ can be seen as points deep inside $K$.
  \begin{definition}[\cite{ellipsoid}]
  	Given a vector $c \in \mathbb{Q}^{n'}$, a rational number $\gamma$ and a rational number $\epsilon > 0$, the Weak Validity problem is to assert either (1) $c^T x \le \gamma + \epsilon$ for all $x \in S(K,-\epsilon)$, or (2) $c^T x \ge \gamma - \epsilon$ for some $x \in S(K,\epsilon)$. Note that the vector $x$ satisfying the second inequality is not required.
  \end{definition}
  \begin{definition}[\cite{ellipsoid}]
  	Given a vector $y \in \mathbb{R}^{n'}$ and $\delta > 0$, the Weak Membership problem is to assert either (1) $y \in S(K,\delta),$ or (2) $y \not \in S(K,-\delta)$.
  \end{definition} 
  Intuitively, in the Weak Membership problem,  it is required to distinguish between the cases when the given point $y$ is far from the polyhedron $K$ (in which case, the algorithm should return $y \not \in S(K,-\delta)$) and   $y$ is deep inside $K$ (which case the algorithm should return $y \in S(K,\delta)$).  If $y$ is near the boundary of $K$, then either output can be returned. Our reduction crucially uses the following result. The notation $<K>$, $<\delta>$ denotes the number of bits used to represent $K$ and $\delta$. 

  \begin{theorem} [Theorem 4.4.4 of \cite{ellipsoid}]
  	\label{4.4.4}
  	Given a weak validity oracle for  $K \subseteq \mathbb{R}^{n'}$  that runs in time polynomial in  $(n',\langle K\rangle,\langle c \rangle)$ and a positive $R$ such that $K \subseteq S(0,R)$ then  the Weak Membership problem for the polyhedron $K$  can be solved in  time polynomial in input size $(n',\langle K\rangle,\langle \delta \rangle)$.
  \end{theorem}

For our problem $K$ is the polytope of linear program Norm-D.
   \begin{equation}
       K := \left \{y \in \mathbb{R}^n : \sum_{i : S \cap T_i \neq \emptyset } y_i \le 0 \quad  \forall S \subseteq [m] ,\quad ||y||_\infty \le 1 \right \}\, . \label{eqn:polytopeK}
   \end{equation}

\paragraph*{Coverage WM $\le_p$ Coverage WV $\le_p$ Coverage Norm Extension.} Coverage Weak Membership is the Weak Membership problem for polytope $K$~\eqref{eqn:polytopeK}. Given a set $\mcD = \{T_1,\dots,T_n\}$ (where $T_i \subseteq [m]$) with weights  $\hat{y}_i$ ($ \hat{y}_i  \in \mathbb{R}$) associated with $T_i$ for all $i \in [n]$  and a $\delta > 0$, the goal in this problem is to assert either $(\hat{y}_1,\dots,\hat{y}_n) \in S(K,\delta)$ or $(\hat{y}_1,\dots,\hat{y}_n) \not \in S(K,-\delta)$.  

Note that Coverage Norm Extension is  the Strong Validity problem for $K$ with $c_i = f_i$. We  show the following lemma (Coverage WV $\le_p$ Coverage Norm Extension).
 
 \begin{lemma}
 	\label{norm-wv}
 	If there is an $\alpha = 2^{poly(n,m)} F^\delta$ efficient approximation algorithm (for any fixed $0 \le \delta < 1$) for Coverage Norm Extension then there is an efficient algorithm for Weak Validity problem for $K$.
 \end{lemma}
   Theorem~\ref{4.4.4} immediately gives Coverage WM $\le_p$ Coverage WV.

\paragraph*{Span WM $\equiv$ Coverage WM.} In fact, we show that Coverage Weak Membership is NP-hard  even for the case when $|T_i| = 2$ for all $i \in [n]$.\footnote{There is a relatively easier proof for unrestricted $d$ by reduction from Set Cover, which we show in the appendix.} The restriction $|T_i|=2$ gives us a graphical representation of the membership problems. We first introduce some notations, which will be used in the remainder. Given a weighted graph  $G = (V,E)$ and a set $S \subseteq V$, the span $E_{G}^+(S)$ and cut $\delta_{G}(S)$ of set $S$ are the set of edges with at least one endpoint and exactly one endpoint in $S$ respectively. 
	We use $w(E_{G}^+(S)), w(\delta_{G}(S))$ and $w(E_{G}(S))$ for the sum of weight of edges with at least one endpoint, exactly one endpoint and both endpoints in $S$ respectively. If the set $S$ is a single vertex $v$ then we use $v$ instead of $\{v\}$. If the graph $G$ is understood from the context we drop the subscript $G$.
	
	
	
		
		Given a set $\mcD = \{T_1,\dots,T_n\}$ $(T_i \subseteq [m])$ with the property that $|T_i| = 2$ for all $i \in [n]$, we construct a weighted graph $G = (V,E)$ as follows: vertex set $V = [m]$ and $\{i,j\} \in E$ ($i,j \in [m]$) iff there exists a $T_k \in \mcD$ such that $T_k = \{i,j\}$. The weight 
		$\hat{y}_k$ associated with $T_k = \{i,j\}$ is now associated to the edge $\{i,j\}$. 	Now the constraint $ \sum_{i : T_i \cap S \neq \emptyset}  y_i \le 0$  (in the polyhedron $K$) translates to  $ \sum_{ e \in E^{+}(S)}  y_e \le 0$ for all $S \subseteq V$. Thus Coverage-Weak-Membership for $|T_i| = 2$ case is equivalent to following problem, which we call \emph{Span Weak Membership}.
		

		 Given a  weighted graph $G = (V,E)$ with  weights $\hat{y}_e$ on the edges  and $\delta > 0$, assert either $\hat{y} = (\hat{y}_e)_{e \in E}$ is in $S(K_s,\delta)$ or $\hat{y}$ is not in $S(K_s,-\delta)$, where 
		 \begin{equation}
		K_s = \left \{\sum_{e \in E^+(S)} y_e \le 0 \quad  \forall S \subseteq V , \,  ||y||_\infty \le 1 \right \} \, .
		\label{eqn:definekprime}
		\end{equation}
		
\paragraph*{Densest-Cut $\le_p$ Cut WM $\le_p$ Span WM.}	We now show that the Span Weak Membership  is NP-Hard thereby showing Coverage Weak Membership is also NP-Hard for the restricted setting with $|T_i| = 2$ for all $i \in [n]$.  We first define Cut Weak Membership.
		
		 Given a weighted graph $G = (V,E)$ with weights $\hat{y}_e$ on the edges  and $\delta > 0$, the goal in Cut Weak Membership is to assert either $\hat{y} = (\hat{y}_e)_{e \in E}$ is in $S(K_c,\delta)$ or $\hat{y}$ is not in $S(K_c,-\delta)$ where 
		 \begin{equation}
		K_c = \left \{\sum_{e \in \delta(S)}  y_e \le 0 \quad  \forall S \in  2^{V}\setminus \emptyset ,  ||y||_\infty \le 1 \right \} \, .
		\label{eqn:definek}
		\end{equation}
		
		Note that in the Cut Weak membership problem, we have  constraints $\sum_{e \in \delta(S)} y_e \le 0 $ instead of $\sum_{e \in E^+(S)} y_e \le 0$ for all $S$. 
		
		
		\begin{lemma}
			\label{densest-cut-span}
	There is a reduction from Densest-Cut to Cut Weak Membership  and from  Cut Weak Membership to Span Weak Membership. Therefore, Coverage Weak Membership is NP-hard even when $d = 2$.
		\end{lemma}

We can now complete the proof of Theorem~\ref{norm-lb}.

			\begin{proof}[Proof of Theorem~\ref{norm-lb}.]
			Suppose there is an efficient $\alpha$-approximation algorithm for  Coverage Norm Extension. Then by Lemma~\ref{norm-wv} there is an efficient algorithm for Weak Validity problem for polytope $K$~\eqref{eqn:polytopeK} and then by Theorem~\ref{4.4.4}  we have an efficient algorithm for Coverage Weak Membership. But by Lemma~\ref{densest-cut-span}, this is not possible unless $P = NP$.
			\end{proof}

		 We here prove  Lemma~\ref{weaker-densest-cut-span}, which is a weaker statement than Lemma~\ref{densest-cut-span} to convey the main ideas.  Recall that in Strong Membership problem,  the goal is to decide if given vector $y$ is in polyhedron $K$. Following our nomenclature, we define the following Strong Membership problems.
		 
		 An instance of  Span Strong Membership and Cut Strong Membership is given by a
		   weighted graph $G = (V,E)$ with  weights $\hat{y}_e$ on the edges, and
		 the goal is to decide if  vector $y = (y_e)_{e \in E}$ is  in $K_s$  and $K_c$ respectively, with $K_s$ and $K_c$ as defined in~\eqref{eqn:definekprime},~\eqref{eqn:definek}. 
		 

		 
		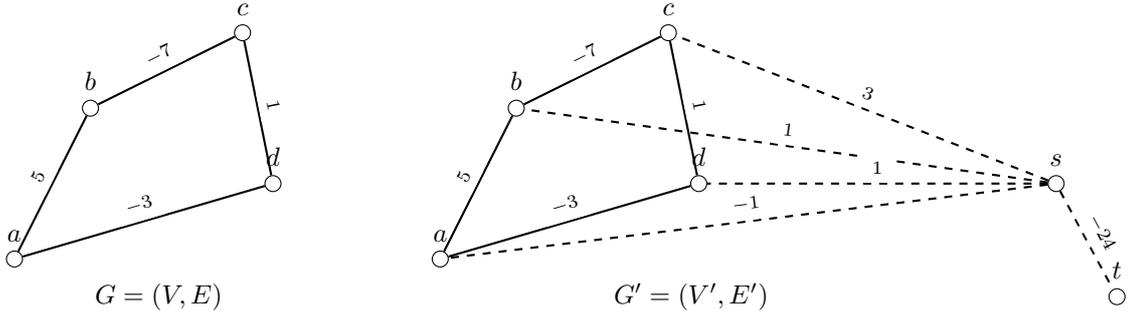
\begin{figure}
			\begin{tikzpicture}
			\vertex[label=$a$](a) at (0,0) {};
			\vertex[label=$b$](b) at (1,2) {};
			\vertex[label=$c$](c) at (3,3) {};
			\vertex[label=$d$](d) at (3.4,1) {};
			\tikzset{EdgeStyle/.style={font=\scriptsize,above,sloped,midway}}
			\Edge[label = $5$](a)(b)
			\Edge[label = $-7$](b)(c)
			\Edge[label = $1$](c)(d)
			\Edge[label = $-3$](d)(a)
			
			\node[] at (1.9,-0.5) { $G = (V,E)$ };
			
			\vertex[label=$a$](a) at (5.6,0) {};
			\vertex[label=$b$](b) at (6.6,2) {};
			\vertex[label=$c$](c) at (8.6,3) {};
			\vertex[label=$d$](d) at (9,1) {};
			\vertex[label=$s$](s) at (13.7,1) {};
			\vertex[label=$t$](t) at (14.5,-0.5) {};
			\tikzset{EdgeStyle/.style={font=\scriptsize,above,sloped,midway}}
			\Edge[label = $5$](a)(b)
			\Edge[label = $-7$](b)(c)
			\Edge[label = $1$](c)(d)
			\Edge[label = $-3$](d)(a)
			\Edge[label = $-1$, style = {dashed}](s)(a)
			\Edge[label = $1$, style = {dashed}](s)(b)
			\Edge[label = $3$, style = {dashed}](s)(c)
			\Edge[label = $1$, style = {dashed}](s)(d)
			\Edge[label = $-24$, style = {dashed}](s)(t)
			
			\node[] at (8.9,-0.5) { $G' = (V',E')$ };
			\end{tikzpicture}
			\caption{Reduction from Cut Strong Membership to Span Strong Membership. The number shown on the edges in $E$ is the weight $y_e$, while on edges in $E'$ is the product of  $L  = 24$ and weight $y'_e$.} \label{fig:strong}
		\end{figure}
		\begin{lemma}
			\label{weaker-densest-cut-span}
			There is a reduction from Densest-Cut to Cut Strong Membership, and from  Cut Strong Membership to Span Strong Membership. 
		\end{lemma}
		\begin{proof}
For the second reduction, the instance of Cut Strong Membership is weighted graph $G = (V,E)$ with  weights $y_e$ on the edges.  We assume $||y||_\infty \le 1$ as otherwise clearly $y \not \in K_c$.
			
Let $L =   2 |E| +  |V| |E|$.			We construct an instance of Span Strong Membership (see Figure \ref{fig:strong}), i.e., graph $G'= (V',E')$ and weights $y'_e$ as follows:
			\[
			V' = V \cup \{s,t\} \, , ~ E' = E \cup \{s,t\} \cup \{v,s\} \quad \forall v \in V \, , ~ y'_e = 
			\left\{
			\begin{array}{ll}
			\frac{y_e}{L}  & \mbox{if } e \in E(G) \\
			-\frac{1}{2L} w(\delta_{G}(v)) & \mbox{if } e = \{v,s\}, v \neq t \\
			-1 & \mbox{if } e = \{s,t\}. \\
			\end{array}
			\right.
			\]
			
\noindent Then $||y'||_\infty \le 1$. 

			Assume $y \not \in K_c$, i.e., there exists  $S \subseteq V$ s.t. $w(\delta_{G}(S)) > 0$. We need to show there exists  $S' \subseteq V'$ s.t. $\sum_{e \in E^{+}(S')} y'_e > 0 $. For $S' = S$, $L \sum_{e \in E^{+}(S')} y'_e = w(E_{G}(S)) + w(\delta_{G}(S)) + \sum_{v \in S} -\frac{1}{2} \cdot w(\delta_{G}(v)) = w(E_{G}(S)) + w(\delta_{G}(S)) - \frac{1}{2} \cdot (2 w(E_{G}(S)) + w(\delta_{G}(S))) = \frac{w(\delta_{G}(S))}{2} > 0$. 
			
			Now assume $y \in K_c$, i.e., $\forall S \subseteq V, w(\delta_{G}(S)) \le 0 $. We need to show $\forall S' \subseteq V', \sum_{e \in E^{+}(S')} y'_e \le 0$. Since $y'_{\{s,t\}} = -1 $ (and $L$ is  sufficiently large), we need to consider only those $S'$ which do not contain either $s$ or $t$. But we have shown that for such $S'$, $\sum_{e \in E^{+}(S')} y'_e = \frac{w(\delta_{G}(S'))}{2L} \le 0$. 
			
			Now we finish the proof by giving a reduction from  Densest-Cut to Cut Strong Membership. Given an  undirected graph $G = (V,E)$  and rational $M$, we want to know if there exists $ S \subset V $ s.t.  $ \frac{\delta_{G}(S)}{|S| |V \setminus S|} > M$.  Consider the complete graph $G' = (V,E')$ where the weight of an edge is $\frac{1 - M}{L}$ if it existed in $E$, and is $-\frac{M}{L}$ otherwise (note that edges may now have positive, negative, or zero weight). Let $L' =  2  \max\{M,|1-M|\}$ be a sufficiently large quantity so that $||\hat{y}||_\infty  < 1$ . It is easy to see that $L \thinspace w(\delta_{G'}(S)) = |\delta_{G}(S)| - M  |S| |V \setminus S|$. Therefore, $\exists S \subset V $ s.t. $w(\delta_{G'}(S)) > 0 \Leftrightarrow \exists S \subset V $ s.t. $\frac{|\delta_{G}(S)|}{|S| |V \setminus S|} > M$. 
		\end{proof}

	\bibliography{partialfn}

	\appendix
		\section{Appendix}
\label{sec:coverageappendix}

 \subsection*{Proof of  Proposition~\ref{succinct}} 
 Consider the polyhedron Extension-P. If the partial function is extendible, then Extension-P is nonempty. Since the variables are non-negative, the polyhedron must have a vertex~\cite{bertsimas1997introduction}, and in particular there is a vertex in which at most $n$ variables $w(S)$ are non-zero. This is  because  the dimension of the problem is $2^m$, hence at a vertex at least $2^m$ constraints must be tight. But then at least $2^m - n$ of constraints $w(S) \ge 0$ must be tight.


 \subsection*{Proof of  Corollary \ref{equality-x}} 
 Consider the polytope 	$P = \{\sum_{I \in \mathcal{I}: v \in I} x_I = 1 \quad \forall v \in  V(G), 0 \le x_I \le 1 \quad \forall  I \in \mathcal{I}  \}$. By the Theorem \ref{frac-colouring}, there exists $x = \{x_I\}_{I \in \mathcal{I}}$ in $P$ such that $\chi^*(G) = \sum_{I \in \mathcal{I}} x_I$. Consider $y = \{y_I\}_{I \in \mathcal{I}}$ given by $y_{\{v\}} = 1$ for all $v \in V(G)$ and $0$ otherwise. Therefore,  $y \in P$ and  $|V(G)| = \sum_{I \in \mathcal{I}} y_I$. Consider $z = \lambda x + (1 - \lambda)y$ where $ \lambda = \frac{|V(G)| - t}{|V(G)| - \chi^*(G) }$. Therefore, $z \in P$ and $\sum_{I \in \mathcal{I}} z_I = \lambda \sum_{I \in \mathcal{I}} x_I + (1 - \lambda) \sum_{I \in \mathcal{I}} y_I = t$.
 
  \subsection*{Proof of  Lemma \ref{coverage-claim}} 
  From Theorem \ref{mobius-inversion}, given a partial function $H$ and $\alpha \ge 1$, there exists a coverage function $f$ satisfying $f_i \le f(T_i) \le \alpha f_i$ for all $ i \in [n]$ iff  the following linear program   is feasible, where the variables are the $W$-coefficients $w(S)$ for all $S \in  2^{[m]} \setminus \emptyset$ : 
  \[f_i \le \displaystyle\sum_{S : S \cap T_i \neq \emptyset } w(S) \le \alpha f_i \quad \forall i \in [n] \]
  \begin{equation*}
  w(S) \ge 0 \quad  \forall S \in  2^{[m]}\setminus \emptyset.
  \end{equation*}
  
  By Farkas' Lemma, it follows that the above linear program is feasible iff the following linear program  is infeasible, with variables $y_i$ and $z_i$ for all $i \in [n]$: \\ 
  \begin{equation}
  \label{inequality-covd-1}
  \alpha \sum_{i = 1}^{n} f_i y_i < \sum_{i = 1}^{n} f_i z_i
  \end{equation}
  \begin{equation}
  \label{inequality-covd-2}
  \sum_{i : S \cap T_i \neq \emptyset } y_i \ge \sum_{i : S \cap T_i \neq \emptyset } z_i \quad  \forall S \in  2^{[m]}\setminus \emptyset 
  \end{equation}
  \[y_i,z_i \ge 0\]. 
  
  Now we proceed towards proving the claim. Suppose $l_i$'s satisfy (\ref{inequality-cov-1}) and (\ref{inequality-cov-2}). Set $y_i$ and $z_i$ as follows: If $l_i \le 0$ then let $y_i = -l_i$ and $z_i = 0$. Else if $l_i > 0$ then let $y_i = 0$ and $z_i = l_i$. It is easy to see that $y_i,z_i \ge 0$ and $l_i = z_i - y_i$ and hence (\ref{inequality-covd-2}) is satisfied by $y_i$'s and $z_i$'s. Further, $\alpha \sum_{i = 1}^{n} f_i y_i = \alpha (\sum_{i: l_i \le 0} f_i y_i + \sum_{i: l_i > 0} f_i y_i) =   -\alpha \sum_{i: l_i \le 0} f_i l_i $ and similarly $ \sum_{i = 1}^{n} f_i z_i =  \sum_{i: l_i > 0} f_i l_i $. Thus (\ref{inequality-covd-1}) is also satisfied by $y_i$'s and $z_i$'s. 
  
  For the other direction observe that if the vector $y = (y_1,..,y_n),z = (z_1,...,z_n) \ge 0$ satisfy (\ref{inequality-covd-1}) and (\ref{inequality-covd-2}) then wlog we can assume for any $i$, the minimum of $y_i$ and $z_i$ is 0 (otherwise we can decrease both $y_i$ and $z_i$ by the minimum of $y_i$ and $z_i$, and $\alpha \ge 1$  allows (\ref{inequality-covd-1}) to remain true). Note that $\sum_{i} f_i y_i = \sum_{i: y_i \le z_i} f_i y_i + \sum_{i: y_i > z_i} f_i y_i = \sum_{i : y_i > z_i} f_i y_i$, since $\min\{y_i,z_i\} = 0$ by the previous observation.  Now suppose $y,z \ge 0$ satisfy (\ref{inequality-covd-1}) and (\ref{inequality-covd-2}). We thus have $\alpha \sum_{i = 1}^{n} f_i y_i < \sum_{i = 1}^{n} f_i z_i \Leftrightarrow \alpha \sum_{y_i > z_i} f_i y_i < \sum_{z_i > y_i} f_i z_i $. Now let $l_i = z_i - y_i$. This makes both (\ref{inequality-cov-1}) and (\ref{inequality-cov-2}) true.	
  
    \subsection*{Proof of  Lemma \ref{computing-c}} 
    	Suppose we are given a weighted bipartite graph $G = (A \cup [m],E)$ with weight $f_v$ on each $v \in A$. Recall that $\kappa$ is the minimum of $\frac{\sum_{w\in R}f_w}{f_v}$ over vertices $v \in A$ and $R \in \mathcal{F}_v$ where  $\mathcal{F}_v = \{R \subseteq A \setminus \{v\} | N(R) \supseteq N(v)\}$ is the set of all $R \subseteq A \setminus \{v\}$ that covers all the neighbours of $v$. 
    	
    	We will use $f(R)$ ($R \subseteq A$) to denote the summation $\sum_{v \in R} f_v$. If $d$ is a constant then for each $v \in A$, we can find minimum of $f(R)$ over all $R \subseteq \mathcal{F}_v$ in $O(n^d)$ time where $n = |A|$. Therefore, by taking the minimum of the above minimum value over all vertices $v \in A$, we get the value of $\kappa$. For general $d$, we use an approximation algorithm for \textsf{Set-Cover} to find, for each vertex $v \in A$,  a set $R'_{v} \in \mathcal{F}_v $ such that $f(R'_{v}) \le  f(R_v) \log d$ where $R_v$ is the optimal set. It can be seen that $\kappa' = \min_{v \in A} \frac{f(R'_{v})}{f_v}$ has the property $\kappa' \le \kappa \log d$. 
     \subsection*{Remaining proof of Theorem \ref{tight-logm-factor} }
      We will show that
      \begin{equation*} \label{prob-inequality}
      Pr[ |N(S) \cap  \mathcal{N}| \ge |N(S) \cap  \mathcal{P}| \quad  \forall S \subseteq [m]] > 0
      \end{equation*}
      Let $\mathcal{S} = \{S \subseteq [m]| |S \cap P| \le 1 \quad \forall P \in \mathcal{P}\}$. Also let $N^-(S) = N(S) \cap  \mathcal{N}$ and $N^+(S) = N(S) \cap  \mathcal{P}$.  Note that $ |N^-(S)| \ge |N^+(S)|$ for all $S \in \mathcal{S}$ implies $ |N^-(S)| \ge |N^+(S)|$  for all $S  \subseteq [m]$.
      
      Given $S \in \mathcal{S}$ of size $s$ and a set $M \subseteq \mathcal{N}$ of size $s$, the probability\footnote{The analysis from here on is similar to existence proof of a  expander on a bipartite graph.} that $N^-(S) \subseteq M$ is $ (1 - 1/\sqrt{m})^{(|\mathcal{N}|-s)s}$. Therefore, $Pr[ |N^-(S)| \ge |N^+(S)| \quad \forall S \subseteq [m]] = 1 - \sum_{S \in \mathcal{S}} \sum_{M \subseteq \mathcal{N}, |M| = |S|} (1 - 1/\sqrt{m})^{(|\mathcal{N}|-s) s} = 1 - \sum_{s = 1}^{\sqrt{m}}  \binom{\sqrt{m}}{s} \sqrt{m}^s \binom{k \sqrt{m} \log m}{s}  (1 - 1/\sqrt{m})^{(k\sqrt{m}\log m-s)s}\\
      > 1 -
      \sum_{s = 1}^{\sqrt{m}}  \frac{(\frac{\sqrt{m}e}{s})^s \sqrt{m}^s (\frac{ke \sqrt{m} \log m}{s})^s}  {e^{ k s \log m}} > 1 - \sum_{s = 1}^{\sqrt{m}}(O(\frac{ m k \sqrt{m} \log m}{ m^k}))^s$ which is strictly greater than $0$ for sufficiently large $k$.
    \subsection*{Proof of  Lemma \ref{norm-wv} }
    	     The instance of weak validity problem are  vector $c \in \mathbb{Q}^{n}$ and  rational numbers $\gamma$ and  $\epsilon > 0$. 
    	     We show that there is a reduction from general Weak Validity problem to Weak Validity problem with instances satisfying $c_i \ge 0$ for all $i \in [n]$.
    	     
    	     Let $N = \{i \in [n]| c_i \le 0\}$. Consider a vector $c'$ such that $c'_i = 0$ for $i \in N$ and $c_i$ otherwise and $\gamma' = \gamma - \sum_{i\in N} |c_i|$.  
    	     If $x$ is in $S(K,\epsilon)$ then clearly $\bar{x}$ defined as $\bar{x}_i = -1$ if $i \in N$ and $x_i$ otherwise, is also in $S(K,\epsilon)$.
    	     If for some $x$ in $S(K,\epsilon)$, we have $(c')^T x \ge \gamma' - \epsilon$ then for  $\bar{x} \in S(K,\epsilon)$,  we have $c^T \bar{x} = \sum_{i\in N} |c_i| + (c')^T x \ge \gamma - \epsilon$.
    	     Also if for all  $x \in S(K,-\epsilon)$,  we have $(c')^T x \le \gamma' + \epsilon$ then $c^T x \le \sum_{i\in N} |c_i| + (c')^T x \le \gamma + \epsilon$. This shows the reduction and hence we assume $c_i \ge 0$ in the instance of Weak Validity problem.


    	     Let $OPT$ and $OPT'$ be the optimal value of Norm-P for $(f_1,\dots,f_n) = (c_1,\dots,c_n)$ and  $(f_1,\dots,f_n) = (L c_1,\dots,L c_n)$ respectively  ($L$ will be chosen later). Obviously $OPT' = L \cdot OPT$. Let the approximation algorithm for Coverage Norm Extension (and hence Norm-P) return $\beta$ for instance $(f_1,\dots,f_n) = (L c_1,\dots,L c_n)$. Let $C = \sum_{i} c_i$. Therefore, $OPT'$ $\le \beta$ $\le OPT' + 2^{poly(n,m)} (L C)^\delta $ $= L \cdot OPT + 2^{poly(n,m)} (L C)^\delta$ and hence $\beta/L \le OPT + \frac{2^{poly(n,m)} (C)^\delta}{L^{1-\delta}}$. We set $L : = \left(\frac{2^{poly(n,m)} (C)^\delta}{2 \epsilon}\right)^{1/1-\delta}$ so that $ \frac{2^{poly(n,m)} (C)^\delta}{L^{1-\delta}} <2 \epsilon$. Note that the number of bits to specify $L$ is polynomial in $\langle c \rangle,\langle \epsilon \rangle, n,m$. Thus, $OPT \le $   $\beta/L \le OPT + 2 \epsilon$. Now if $\gamma + \epsilon \le \beta/L$ then for the optimal solution $x^* \in K$, $c^T x^* = OPT \ge \frac{\beta}{L} - 2 \epsilon$ $\ge \gamma - \epsilon$. If $\gamma + \epsilon \ge \beta/L$ then for all $x$ in $K$ (and hence $S(K,-\epsilon)$), we have $c^T x \le OPT \le \beta/L$ $\le \gamma+\epsilon$. Since at least one of these two conditions must hold, the conditions of weak validity problem can be correctly asserted.

    In the following proofs, for any vector $y$, recall that we  use $||y||_{\infty}$ for $\max_{i} |y_i|$ and $||\hat{y} - y||$ for the Euclidean distance between $\hat{y}$ and $y$.
    We will frequently use the fact that the distance of a point $x_0$ from the hyperplane $w^T  x + b = 0$ is equal to $\frac{|w^T x_0 + b|}{||w||}$.

    \subsection*{Proof of  hardness of Coverage Weak Membership for unrestricted $d$} 
    	We give a reduction from Set Cover. An instance of Set Cover  consists of a universe of $n'$ elements $\mcU = \{u_1, \dots, u_{n'}\}$,  a family of sets $\mcF = \{S_1, \dots, S_{m'}\}$ such that $S_i \subseteq \mcU$ for each $i \in [m']$, and a positive integer $k$. We need to determine if there exists $\mcF' \subseteq \mcF$ and $|\mcF'| \le k$ such that $\mcF'$ covers the universe $\mcU$, i.e., $\cup_{S \in \mcF'} S = \mcU$.
    	
    	The instance of Coverage-Weak-Membership requires  a 
    	set $\mcD = \{T_1,\dots,T_n\}$ where each $T_i \subseteq [m]$, real values $\hat{y}_i$  associated with each $T_i$ and $\delta > 0$.  The goal is to assert either the point $\hat{y} = (\hat{y}_i)_{i \in [n]}$ is in $S(K,\delta)$ or is not in $S(K,-\delta)$. Given an instance of Set Cover, we construct an instance of Coverage Weak Membership as follows. We set $m = m'$, the number of sets in $\mcF$, hence each set $T_i$ can be thought of as a subset of $\mcF$. The instance has $n = m' + 1+ m'$ sets in $\mcD$, as follows (the value of $L$ will be chosen later):
    	
    	\begin{enumerate}
    	    \item The first $m'$ sets correspond to sets in $\mcF$: $T_i = \{i\}$ with $\hat{y}_i = -\frac{1}{L}$ 
    	    \item Set $T_{m'+1} = \{1,\dots,m'\}$  with $\hat{y}_{m'+1} = \frac{k -kn' + \frac{1}{2}}{L}$
    	    \item For each of the $n'$ elements in $\mcU$, there is a set in $\mcD$: For each $u_i \in \mcU$,  $\mcD$ contains the set $\{j \in [m] | S_j \text{ contains } u_i \}$ with $\hat{y}$-value $\frac{k}{L}$.  
    	\end{enumerate}. 
    	    
    	    The value of $L$ is set to $2 \max\{-1, k -kn' + \frac{1}{2},k\} = kn' -k - 1/2$ (to ensure $||\hat{y}||_\infty =  \frac{1}{2}$ in view of the constraint $||y||_\infty \le 1$ of $K$). We set $\delta$ to be $\frac{1}{4L\sqrt{n}}$.
    	
    		Suppose that this is a YES instance of the Set Cover problem, i.e., there exists a $\mcF' \subseteq \mcF$ and $|\mcF'| \le k$ such that $\mcF'$ covers the universe $\mcU$. 
    	In our instance,  if we take the set of elements $S = \{j \in [m']| S_j \in \mcF'\}$ then this intersects with at most $k$ sets of the first kind, set $T_{m'+1}$, and all sets of the third kind. Hence, $\sum_{i : T_i \cap S \neq \emptyset}  \hat{y}_i \ge \frac{-k+(k - k n' +\frac{1}{2}) + kn'}{L} = \frac{1}{2L}$. So the point $\hat{y}$ violates the constraint $\sum_{ i : S \cap T_i \neq \emptyset}  y_i \le 0$ and is at least $\frac{1}{2L\sqrt{n}}$ distance  away from the corresponding hyperplane $\sum_{ i : S \cap T_i \neq \emptyset} y_i = 0$. Therefore, $||\hat{y} - y|| \ge \frac{1}{2L\sqrt{n}}$ for all $y \in K$. 
    	Hence $\hat{y} \not \in S(K,\delta)$ as $\delta <\frac{1}{2L\sqrt{n}}$.
    	
    	Now suppose that this is a NO instance of the Set Cover problem.
    	For any $S \subseteq 2^{[m']}\setminus \emptyset$, if $|S| \ge k+1$ then $\sum_{i : T_i \cap S \neq \emptyset}  \hat{y}_i  \le \frac{-k-1 + (k - kn'+\frac{1}{2})+ kn'}{L}  = -\frac{1}{2L}$ and if $|S| \le k$ then $ \sum_{i : T_i \cap S \neq \emptyset}  \hat{y}_i  \le \frac{(n' -1)k - 1 + (k - kn'+\frac{1}{2})}{L}  = -\frac{1}{2L}$, since any set of size strictly less than $k+1$ covers at most $n'-1$ elements of universe $\mcU$.  Thus the point $\hat{y}$ satisfies  the constraints $\sum_{ i : T_i \cap S \neq \emptyset}  y_i \le 0$ for all $S$ and is at least $\frac{1}{2L \sqrt{n}}$  distance away from all the hyperplanes $\sum_{i : T_i \cap S \neq \emptyset} y_i = 0$.  Also, for any $y \in S(\hat{y},\delta)$, we have  $||y||_\infty - ||\hat{y}||_\infty \le||y - \hat{y}||_\infty \le ||y - \hat{y}|| \le \delta$. So $||y||_\infty \le   \delta +   1/2 \le 1$ (as $\delta < \frac{1}{2}$). 
    	Thus $\hat{y} \in S(K,-\delta)$.
    	
    	\subsection*{Proof of Lemma~\ref{densest-cut-span}} 
    
    
    
    Recall the definitions of Span Weak Membership, Cut Weak Membership and Densest Cut: 
    
    \begin{enumerate}
    \item Given a  weighted graph $G = (V,E)$ with  weights $\hat{y}_e$ on the edges  and $\delta > 0$,
    \begin{enumerate}
        \item The goal in Span Weak Membership is to assert either $\hat{y} = (\hat{y}_e)_{e \in E}$ is in $S(K_s,\delta)$ or $\hat{y}$ is not in $S(K_s,-\delta)$ where    
    
    \[ K_s = \left \{\sum_{e \in E^+(S)} y_e \le 0 \quad  \forall S \in  2^{V}\setminus \emptyset ,  ||y||_\infty \le 1 \right \}, \]
    
    
    \item The goal in Cut Weak Membership is to assert either $\hat{y} = (\hat{y}_e)_{e \in E}$ is in $S(K_c,\delta)$ or $\hat{y}$ is not in $S(K_c,-\delta)$ where
     \[ K_c = \left \{\sum_{e \in \delta(S)}  y_e \le 0 \quad  \forall S \in  2^{V}\setminus \emptyset ,  ||y||_\infty \le 1 \right \}. \]
\end{enumerate}    
    Note that in the Cut Weak membership, we have  constraints $\sum_{e \in \delta(S)} y_e \le 0 $ instead of $\sum_{e \in E^+(S)} y_e \le 0$ for all $S$. 
    
    \item In the Densest-Cut problem, given a  graph $G = (V,E)$  and a positive rational $M$, the goal is to decide if there exist a set $S \subset V $ s.t. $ \frac{|\delta(S)|}{|S| |V \setminus S|} \ge M$. 
    
    \end{enumerate}

    The Densest-Cut is known to be NP-Hard \cite{densest-cut}. Note that  $ \frac{|\delta(S)|}{|S| |V \setminus S|}$ called the density of cut $(S,V \setminus S)$  can  take  values only from $\left\{\frac{r}{s(|V|-s)} | 1 \le r \le |E|, 1 \le s \le |V|-1, r,s \in \mathbb{Z}_+\right\}$. Thus there are only polynomially many  possible values of cut densities. We will use this fact in our proof.

    	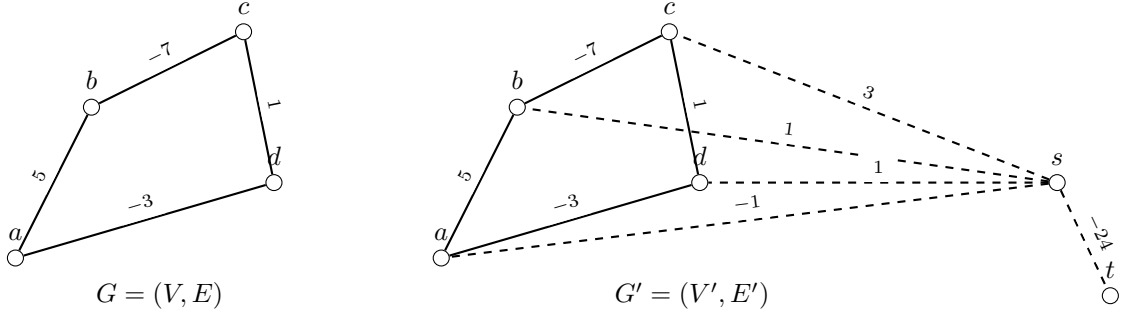
\begin{figure}
    		\begin{tikzpicture}
    		\vertex[label=$a$](a) at (0,0) {};
    		\vertex[label=$b$](b) at (1,2) {};
    		\vertex[label=$c$](c) at (3,3) {};
    		\vertex[label=$d$](d) at (3.4,1) {};
    		\tikzset{EdgeStyle/.style={font=\scriptsize,above,sloped,midway}}
    		\Edge[label = $5$](a)(b)
    		\Edge[label = $-7$](b)(c)
    		\Edge[label = $1$](c)(d)
    		\Edge[label = $-3$](d)(a)
    		
    		\node[] at (1.9,-0.5) { $G = (V,E)$ };
    		
    		\vertex[label=$a$](a) at (5.6,0) {};
    		\vertex[label=$b$](b) at (6.6,2) {};
    		\vertex[label=$c$](c) at (8.6,3) {};
    		\vertex[label=$d$](d) at (9,1) {};
    		\vertex[label=$s$](s) at (13.7,1) {};
    		\vertex[label=$t$](t) at (14.4,-0.5) {};
    		\tikzset{EdgeStyle/.style={font=\scriptsize,above,sloped,midway}}
    		\Edge[label = $5$](a)(b)
    		\Edge[label = $-7$](b)(c)
    		\Edge[label = $1$](c)(d)
    		\Edge[label = $-3$](d)(a)
    		\Edge[label = $-1$, style = {dashed}](s)(a)
    		\Edge[label = $1$, style = {dashed}](s)(b)
    		\Edge[label = $3$, style = {dashed}](s)(c)
    		\Edge[label = $1$, style = {dashed}](s)(d)
    		\Edge[label = $-24$, style = {dashed}](s)(t)
    		
    		\node[] at (8.9,-0.5) { $G' = (V',E')$ };
    		\end{tikzpicture}
    		\caption{Reduction from Cut-Weak-Membership to Span-Weak-Membership. The number shown on the edges in $E$ is the weight $y_e$, while on edges in $E'$ is  product of  $L  = 48$ and weight $y'_e$.} \label{cut-span}
    	\end{figure}
    	
    	\begin{lemma}
    		There is a reduction from Cut Weak Membership to Span Weak Membership.
    	\end{lemma}
    	\begin{proof}
    		Our goal in Cut Weak Membership, given a graph a $G = (V,E)$ with  weights $\hat{y}_e$ on edges and $\delta > 0$, is to assert either  $\hat{y} = (\hat{y}_e)_{e \in E}$ is in $S(K_c,\delta)$ or $\hat{y}$ is not in $S(K_c,-\delta)$. If the point $\hat{y}$ violates   the constraint   $ ||y||_\infty \le 1 $    of $K_c$ then it can be asserted that $\hat{y}$ is not in $S(K_c,-\delta)$. So we assume $ ||\hat{y}||_\infty \le 1 $.   Given this assumption, we have   $ w(\delta_{G}(v)) \le |E|$.
    		
    		We construct an instance of Span-Weak-Membership (see Figure \ref{cut-span}), i.e., graph $G'= (V',E')$,  $\hat{y}'_e$ and $\delta'$ as follows (the values of $B$ and $L$ will be set later): \\ 
    		$V' = V \cup \{s,t\}$\\
    		$E' = E \cup \{\{s,t\}\} \cup \{\{v,s\}\} \quad \forall v \in V', v \neq \{s,t\}$ \\
    		$\hat{y}'_e = 
    		\left\{
    		\begin{array}{ll}
    		\frac{y_e}{L}  & \mbox{if } e \in E \\
    		-\frac{\frac{1}{2} w(\delta_{G}(v))}{L} & \mbox{if } e = \{v,s\}, v \neq t \\
    		-\frac{B}{L} & \mbox{if } e = \{s,t\}. \\
    		\end{array}
    		\right.
    		$\\
    		The value of $B$ is set to $ 2 |E| +  |V| |E|$ so that   $\sum_{e \in E_{G'}^+(S)}  \hat{y}'_e \le \frac{-B + |E| + 1/2 |V| |E|}{L} \le 0$ for all $S$ containing either $s$ or $t$.
    		Further, $L = 2B$ so that $||\hat{y}'||_\infty = 1/2$ where $\hat{y}' = (\hat{y}'_e)_{e \in E'}$. Finally we choose
    		$\delta' = \frac{1}{2}\min \left\{\frac{\sqrt{|E|} \delta}{2L \sqrt{|E'|}}, \frac{|E| + 1/2 |V| |E|}{ \sqrt{ |E'| L}},\frac{1}{2}\right\}$.

    		


    		\begin{claim}\label{spanhalfcut}
    			For all $S \subseteq V$,  $w(E_{G'}^+(S)) =  \frac{w(\delta_{G}(S))}{2L}$. 
    		\end{claim}
    		\begin{proof}
    			This is because  
    			
    			\[ L \, w(E_{G'}^+(S)) = L \, \sum_{e \in E_{G'}^{+}(S)} w'_e = w(E_{G}(S)) + w(\delta_{G}(S)) + \sum_{v \in S} -\frac{1}{2} \, w(\delta_{G}(v))\, , \]
    			
    			\noindent and since $w( \delta_G(v))$ counts edges in $E_G(S)$ twice and edges in $\delta_G(S)$ once,
    			
    			\[ L \, w(E_{G'}^+(S)) = w(E_{G}(S)) + w(\delta_{G}(S)) - \frac{1}{2} \cdot (2 w(E_{G}(S)) + w(\delta_{G}(S))) = \frac{w(\delta_{G}(S))}{2} \, .\]
    		\end{proof}
    		
    		Suppose the algorithm for Span Weak Membership asserts that the point $\hat{y}'$ is in $S(K_s,\delta')$. 
    		If 	$\hat{y}'$  satisfies all the constraints $\sum_{e \in E_{G'}^+(S)}  y_e \le 0$   for all $S \in  2^{V}\setminus \emptyset$ then the point $\hat{y}$ must satisfy  all the constraints $ \sum_{e \in \delta_{G}(S)}  y_e \le 0$ for all $S \in  2^{V}\setminus \emptyset$ (because by the  Claim \ref{spanhalfcut} $w(\delta_{G}(S)) = 2L \cdot w(E_{G'}^+(S)) $) and hence $\hat{y} \in K_c$ . Thus $\hat{y} \in S(K_c,\delta)$. Now suppose $\hat{y}'$ violates a constraint $ \sum_{e \in E_{G'}^+(R)}  y_e \le 0 $ for some $R \in  2^{V}\setminus \emptyset$. Since $\hat{y}' \in S(K_s,\delta')$, 
    		it is at most $\delta'$ distance away from   the hyperplanes corresponding to the violated constraints. Therefore, we have $w(E_{G'}^+(R)) = \sum_{e \in E_{G'}^+(R)}  \hat{y}'_e\le \delta' \sqrt{|E'|}$. 
    		By  Claim \ref{spanhalfcut}, $w(\delta_{G}(R)) \le 2L \delta' \sqrt{|E'|}$. Therefore, the point $\hat{y}$ is at most $\frac{2L \delta' \sqrt{|E'|}}{\sqrt{|E|}}$ distance from $K_c$. Since $\delta' < \frac{\sqrt{|E|} \delta }{2L  \sqrt{|E'|}}$, so $\hat{y} \in S(K_c,\delta)$.

    		Suppose the algorithm for Span Weak Membership problem asserts that the point $\hat{y}'$ is not in $S(K_s,-\delta')$.
    		If 	$\hat{y}'$ violates  a constraint $\sum_{e \in E_{G'}^+(S)}  y_e \le 0$ for some $S \in 2^{V}\setminus \emptyset$ then the point $\hat{y}$ also violates $ \sum_{e \in \delta_{G}(S)}  y_e \le 0$ for $S$ (by  Claim \ref{spanhalfcut}). Hence, it can be asserted that 		$\hat{y}$ is not in $S(K_c,-\delta)$. So now assume that $\hat{y}'$  satisfies all the constraints $\sum_{e \in E_{G'}^+(S)}  y_e \le 0$ for all $S \in 2^{V}\setminus \emptyset$. Also, as shown earlier,  $\hat{y}'$ satisfies  the other constraints of $K_s$. Since  $\hat{y}'$ is in $K_s$ but  not in $S(K_s,-\delta')$,  some $y \in S(\hat{y}',\delta')$ must have distance $< \delta'$ from some hyperplane of $K_s$.
    		The distance of $\hat{y}'$ from the hyperplane   $ \sum_{e \in E_{G'}^+(S)}  y_e = 0$ for $S$ containing $s$ or $t$ is at least $\frac{|-B + |E| + 1/2 |V| |E'||}{\sqrt{ |E'| L}} = \frac{|E| + 1/2 |V| |E'|}{ \sqrt{ |E'| L}} > \delta'$. 
    		Also for any $y \in S(\hat{y}',\delta')$, we have  $||y||_\infty - ||\hat{y}'||_\infty \le||y - \hat{y}||_\infty \le ||y - \hat{y}'||$. So $||y||_\infty \le   \delta +   1/2 \le 1$ for all $y \in S(\hat{y}',\delta')$.
    		Therefore, it must be the case that  distance of $\hat{y}'$ from the hyperplane $ \sum_{e \in E_{G'}^+(S)}  y_e = 0$ for some $S \in 2^{V}\setminus \emptyset$ is  $< \delta'$.
    		By Claim \ref{spanhalfcut},  the distance of the point $\hat{y}$ from the  hyperplane $ \sum_{e \in \delta_{G}(S)} y_e = 0$ is at most $ \frac{2L \sqrt{|E'|} \delta'}{\sqrt{|E|}} < \delta$. Hence, $\hat{y}$ is not in $S(K_c,-\delta)$. \end{proof}
    		

    	Now we finish the proof by giving reduction from  Densest-Cut to Cut Weak Membership.
    	\begin{lemma}
    		There is a reduction from Densest-Cut  to Cut Weak Membership.
    	\end{lemma}
    	
    	\begin{proof}
    		In the  Densest Cut problem, a  graph $G = (V,E)$ and a positive rational $M$ are given  and the goal is to determine if there exists a set $ S \subset V $ s.t. the density of the cut $(S,V \setminus S)$ is at least $M$, i.e.,    $ \frac{|\delta_{G}(S)|}{|S| |V \setminus S|} \ge M$. Let $M = \frac{p}{q}$ for positive integer $p,q$.
    		
    		Given the graph $G = (V,E)$ and $M$,  the instance of the Cut Weak Membership  is  a complete graph $G' = (V,E')$ (so $|E'| = \frac{|V|(|V|-1)}{2} $), weight $\hat{y}_e$  on each edge $e \in E'$ such that $\hat{y}_e$ is $\frac{1 - M}{L} $ if it existed in $E$ and  $\frac{-M}{L}$ otherwise  and $\delta = \frac{1}{2} \min \{ \frac{1}{2},\frac{t}{\sqrt{|E'|}}\}$.  Let $\hat{y} = (\hat{y}_e)_{e \in E'}$. 
    		We set   $L$ to $ 2  \max\{M,|1-M|\}$ so that $||\hat{y}||_\infty = 1/2$ and $t$ to $ \frac{1}{qL}$ .
    		
    		It is easy to see that $w(\delta_{G'}(S)) = \frac{1}{L} (|\delta_{G}(S)| - M  |S| |V \setminus S|)$. Therefore, $\exists S \subset V $ s.t. $w(\delta_{G'}(S)) \ge 0 \Leftrightarrow \exists S \subset V $ s.t. $\frac{|\delta_{G}(S)|}{|S| |V \setminus S|} \ge M$.
    		
    		
    		Since $M$ is equal to $\frac{p}{q}$ for some $p,q \in \mathbb{Z}_+$, therefore the weight of an edge is either $\frac{q - p}{qL}$ or $\frac{-p}{qL}$. So  if a cut value $w(\delta_{G'}(S))$ is strictly positive for any $S$ then $w(\delta_{G'}(S))$ must be at least $\frac{1}{qL} = t$. Similarly, if $w(\delta_{G'}(S)) < 0$ then we have $w(\delta_{G'}(S)) \le - t$.
    		
    		Now suppose an algorithm for Cut-Weak-Membership asserts $\hat{y}$ is in $S(K_c,\delta)$. 
    		Thus for all $S$, $w(\delta_{G'}(S)) = \sum_{e \in \delta_{G'}(S)}  \hat{y}_e  \le  \sqrt{|E'|} \delta$.
    		Since $\delta < \frac{t}{\sqrt{|E'|}}$, so it must be the case that for all $S$, $w(\delta_{G'}(S)) \le 0$. This implies that for all $S$, the cut density $\frac{|\delta_{G}(S)|}{|S| |V \setminus S|} \le M$.
    		
    		Now suppose the algorithm for Cut Weak Membership asserts that $\hat{y}$ is not in $S(K_c,-\delta)$.  If  $\hat{y} \not \in K_c$ (and since $||\hat{y}||_\infty \le 1$) then clearly there  exists a set $S$ such that $w(\delta_{G'}(S))  > 0$. This implies there is a cut $(S,V \setminus S)$ with density $\frac{|\delta_{G}(S)|}{|S| |V \setminus S|} > M$. Now assume $\hat{y}$ is in $K_c$.
    		Now for any $y \in S(\hat{y},\delta)$, we have  $||y||_\infty - ||\hat{y}||_\infty \le||y - \hat{y}||_\infty \le ||y - \hat{y}|| \le  \delta$. So $||y||_\infty \le   \delta +   1/2 < 1$.
    		So   there must exist a   hyperplane $\sum_{e \in \delta_{G'}(S)} y_e = 0$ for some $S$ with at most $\delta$ distance from $\hat{y}$.  Therefore, there exist a set $S$ with $0 \ge w(\delta_{G'}(S)) \ge -\sqrt{|E'|} \delta$. Since $\delta <  \frac{t}{\sqrt{|E'|}}$, this means $w(\delta_{G'}(S)) = 0$ and hence density of cut $(S,V \setminus S)$ is $M$. Thus, if an algorithm for Cut Weak Membership asserts that $\hat{y}$ is not in $S(K_c,-\delta)$ then there exists a cut with density at least $M$.
    		
    		Therefore, assuming an efficient algorithm for Cut Weak Membership, it can be determined if there exists a cut with density at least $M$ or all cuts have density at most $M$. However, the goal in  Densest Cut is to determine if there is a cut with  density $ \ge M$ or all cuts have density strictly less than $M$. But  since the density can take only polynomial number of values  $\left\{\frac{r}{s(|V|-s)} | 1 \le r \le |E|, 1 \le s \le |V|-1, r,s \in \mathbb{Z}_+\right\}$ (as noted before), by using at most two oracle calls to the Cut-Weak-Membership problem we can solve the original problem.
    	\end{proof}




\end{document}